\documentclass[10pt,a4paper]{article}

\usepackage{lmodern}
\usepackage[utf8]{inputenc}
\usepackage[T1]{fontenc}

\newcommand{\fullversion}[1]{#1}

\usepackage[affil-it]{authblk}

\usepackage{fullpage}
\usepackage{tikz}
\usetikzlibrary{positioning}
\usetikzlibrary{calc}
\tikzstyle{fptbox}=[line width=0.5mm,rectangle, minimum height=.8cm,fill=white!70,rounded corners=1mm,draw]
\tikzstyle{fptedge}=[line width=0.5mm,->]
\newcommand{\tworows}[2]{\begin{tabular}{c}{#1}\\[-1mm]{#2}\end{tabular}}

\usepackage{todonotes}
\usepackage{fullpage}

\title{On the complexity of \BD{} and \MP{} in digraphs\footnote{A preliminary version of this paper was accepted for presentation at the conference IWOCA'20 and will appear as~\cite{confversion}.}}

\author[1,2]{Florent Foucaud\footnote{Partially funded by the IFCAM project ''Applications of graph homomorphisms'' (MA/IFCAM/18/39) and by the ANR project HOSIGRA (ANR-17-CE40-0022).}}
\author[2,3]{Benjamin Gras}
\author[2]{Anthony Perez}
\author[4]{Florian Sikora\footnote{Partially funded by the project ESIGMA (ANR-17-CE23-0010)}}

\affil[1]{Univ. Bordeaux, Bordeaux INP, CNRS, LaBRI, UMR5800, F-33400 Talence, France}
\affil[2]{Univ. Orl\'eans, INSA Centre Val de Loire, LIFO EA 4022, F-45067 Orl\'eans, France}
\affil[3]{Universität Trier, Fachbereich 4, Informatikwissenschaften, Germany}
\affil[4]{Univ. Paris-Dauphine, PSL University, CNRS, LAMSADE, 75016 Paris, France}

\usepackage{amssymb}
\usepackage{amsmath}
\usepackage{amsthm}
\usepackage{graphicx}
\usepackage{hyperref}
\usepackage{enumerate}

\newcommand{\size}[1]{\left|#1\right|}
\newcommand{\BD}{{\sc Broadcast Domination}}
\newcommand{\BDS}{\BD}
\newcommand{\MP}{{\sc Multipacking}}
\newcommand{\MDS}{{\sc Multicolored Dominating Set}}
\newcommand{\MIS}{{\sc Multicolored Independent Set}}
\newcommand{\FPTR}{FPT-reduction}
\newcommand{\PTPT}{polynomial time and parameter transformation}
\renewcommand{\geq}{\geqslant}
\renewcommand{\leq}{\leqslant}
\DeclareMathOperator{\mpn}{mp}
\DeclareMathOperator{\mfd}{mfd}

\DeclareMathOperator{\ecc}{ecc}

\newtheorem{theorem}{Theorem}
\newtheorem{lemma}[theorem]{Lemma}
\newtheorem{corollary}[theorem]{Corollary}
\newtheorem{claim}[theorem]{Claim}
\newtheorem{proposition}[theorem]{Proposition}

\newtheorem*{remark}{Remark}

\newenvironment{proofclaim}{
	\noindent \emph{Proof.}
}{%
	\hfill $\diamond$ \\
}

\newcommand{\Pb}[4]{%
\begin{center}
  \begin{tabular}{|l|}%
  \hline
    \begin{minipage}[c]{\textwidth}
      \smallskip%
      \par\noindent%
      #1%
      \par\noindent%
      $\bullet$
      \textbf{\textsf{Input}}: #2% 
      \par\noindent%
      $\bullet$
      \textbf{\textsf{#4}}: Does there exist %
      #3?% 
      \smallskip%
      \par\noindent%
    \end{minipage}
  \\\hline
  \end{tabular}%
\end{center}
}%

\begin{document}

\maketitle

\begin{abstract}
We study the complexity of the two dual covering and packing distance-based problems \BD{} and \MP{} in digraphs. A \emph{dominating broadcast} of a digraph $D$ is a function $f:V(D)\to\mathbb{N}$ such that for each vertex $v$ of $D$, there exists a vertex $t$ with $f(t)>0$ having a directed path to $v$ of length at most $f(t)$. The cost of $f$ is the sum of $f(v)$ over all vertices $v$. A \emph{multipacking} is a set $S$ of vertices of $D$ such that for each vertex $v$ of $D$ and for every integer $d$, there are at most $d$ vertices from $S$ within directed distance at most $d$ from $v$. The maximum size of a multipacking of $D$ is a lower bound to the minimum cost of a dominating broadcast of $D$. Let \BD{} denote the problem of deciding whether a given digraph $D$ has a dominating broadcast of cost at most $k$, and \MP{} the problem of deciding whether $D$ has a multipacking of size at least $k$. It is known that \BD{} is polynomial-time solvable for the class of all undirected graphs (that is, symmetric digraphs), while polynomial-time algorithms for \MP{} are known only for a few classes of undirected graphs. We prove that \BD{} and \MP{} are both NP-complete for digraphs, even for planar layered acyclic digraphs of small maximum degree. Moreover, when parameterized by the solution cost/solution size, we show that the problems are respectively W[2]-hard and W[1]-hard. We also show that \BD{} is FPT on acyclic digraphs, and that it does not admit a polynomial kernel for such inputs, unless the polynomial hierarchy collapses to its third level. In addition, we show that both problems are FPT when parameterized by the solution cost/solution size together with the maximum out-degree, and as well, by the vertex cover number.  %
Finally, we give for both problems polynomial-time algorithms for some subclasses of acyclic digraphs.
\end{abstract}

\section{Introduction}

We study the complexity of the two dual problems \BD{} and \MP{} in digraphs. These concepts were previously studied only for undirected graphs (which can be seen as \emph{symmetric} digraphs, where for each arc $(u,v)$, the symmetric arc $(v,u)$ exists). Unlike most standard packing and covering problems, which are of local nature, these two problems have more global features since the covering and packing properties are based on arbitrary distances. This difference makes them algorithmically very interesting. % 

\medskip\noindent\textbf{Broadcast domination.}
Broadcast domination is a concept modeling a natural covering problem in telecommunication networks: imagine we want to cover a network with transmitters placed on some nodes, so that each node can be reached by at least one transmitter. Already in his book in 1968~\cite{liu1968}, Liu presented this concept, where transmitters could broadcast messages but only to their neighboring nodes.
It is however natural that a transmitter could broadcast information at distance greater than one, at the price of some additional power (and cost).
In this setting, for a given non-zero integer cost $d$, a transmitter placed at node $v$ covers all nodes within radius $d$ from its location. If the network is directed, it covers all nodes with a directed path of length at most $d$ from $v$. 
For a feasible solution, the function $f:V(G)\to\mathbb{N}$ assigning its cost to each node of the graph $G$ (a cost of zero means the node has no transmitter placed on it) is called a \emph{dominating broadcast} of $G$, and the total cost $c_f$ of $f$ is the sum of the costs of all vertices of $G$. 
The \emph{broadcast domination number} $\gamma_b(G)$ of $G$ is the smallest cost of a dominating broadcast of $G$. 
When all costs are in $\{0,1\}$, this notion coincides with the well-studied \textsc{Dominating Set} problem. 
The concept of broadcast domination was introduced in 2001 (for undirected graphs) by Erwin in his doctoral dissertation~\cite{Erwin2001} (see also~\cite{DEH+06,Erwin04} for some early publications on the topic), in the context of advertisement of shopping malls -- which could nowadays be seen as targeted advertising via "influencers" in social networks.
Note that in these contexts, directed arcs make sense since the advertisement or the influence is directed towards someone.
The associated computational problem is as follows.

\Pb{\BDS{}}
{A digraph $D=(V,A)$, an integer $k \in \mathbb{N}$.}
{a dominating broadcast of $D$ of cost at most $k$}
{Question}

\medskip\noindent\textbf{Multipacking.}
The dual notion for \BD, studied from the linear programming viewpoint, was introduced in~\cite{BMT13,TeshimaMasterThesis} and called \emph{multipacking}. A set $S$ of vertices of a (di)graph $G$ is a \emph{multipacking} if for every vertex $v$ of $G$ and for every possible integer $d$, there are at most $d$ vertices from $S$ at (directed) distance at most $d$ from $v$. The \emph{multipacking number} $\mpn(G)$ of $G$ is the maximum size of a multipacking in $G$. Intuitively, if a graph $G$ has a multipacking $S$, any dominating broadcast of $G$ will require to have cost at least $|S|$ to cover the vertices of $S$. Hence the multipacking number of $G$ is a lower bound to its broadcast domination number~\cite{BMT13}. Equality holds for many graphs, such as strongly chordal graphs~\cite{BMY19} and $2$-dimensional square grids~\cite{grids}. For undirected graphs, it is also known that $\gamma_b(G) \leq 2\mpn(G)+3$~\cite{BBF19} and it is conjectured that the additive constant can be removed. %
Consider the following computational problem.

\Pb{\MP{}}
{A digraph $D=(V,A)$, an integer $k \in \mathbb{N}$.}
{a multipacking $S \subseteq V$ of $D$ of size at least $k$}
{Question}

\noindent\textbf{Known results.}
In contrast with most graph covering problems, which are usually NP-hard, Heggernes and Lokshtanov designed in~\cite{HL06} (see also~\cite{LMasterThesis}) a sextic-time algorithm for \BD{} in undirected graphs. This intriguing fact has motivated research on further algorithmic aspects of the problem. 
For general undirected graphs, no faster algorithm than the original one is known. A quintic-time algorithm exists for undirected series-parallel graphs~\cite{Blair}. 
An analysis of the algorithm for general undirected graphs gives quartic time when it is restricted to chordal graphs~\cite{HL06,HS12}, and a cubic-time algorithm exists for undirected strongly chordal graphs~\cite{BMY19}. 
The problem is solvable in linear time on undirected interval graphs~\cite{chang2010} and undirected trees~\cite{BMY19,DBLP:journals/networks/DabneyDH09} (the latter was extended to undirected block graphs~\cite{HS12}). Note that when the dominating broadcast is required to be upper-bounded by some fixed integer $p\geq 2$, then the problem becomes NP-Complete~\cite{CaceresHMPP18} (for $p=1$ this is \textsc{Dominating Set}).% 

Regarding \MP, to the best of our knowledge, its complexity is currently unknown, even for undirected graphs (an open question posed in~\cite{TeshimaMasterThesis,YangMasterThesis}). However, there exists a polynomial-time $(2+o(1))$-approximation algorithm for all undirected graphs~\cite{BBF19}. \MP{} can be solved with the same complexity as \BD{} for undirected strongly chordal graphs, see~\cite{BMY19}. Improving upon previous algorithms from~\cite{MT17,TeshimaMasterThesis}, the authors of~\cite{BMY19} give a simple linear-time algorithm for undirected trees.

\medskip\noindent\textbf{Our results.}
We study \BD{} and \MP{} for directed graphs (digraphs), which form a natural setting for not necessarily symmetric telecommunication networks. %
In contrast with undirected graphs, we show that \BD{} is NP-complete, even for planar layered acyclic digraphs (defined later) of maximum degree~$4$ and maximum finite distance~$2$. %
This holds for \MP, even for planar layered acyclic digraphs of maximum degree~$3$ and maximum finite distance~$2$, or acyclic digraphs with a single source and maximum degree~$5$. %
Moreover, when parameterized by the solution cost/solution size, we prove that \BD{} is W[2]-hard (even for digraphs of maximum finite distance~$2$ or bipartite digraphs of maximum finite distance~$6$ %
without directed $2$-cycles) and \MP{} is W[1]-hard (even for digraphs of maximum finite distance~$3$). %
On the positive side, we show that \BD{} is FPT on acyclic digraphs (DAGs for short) but does not admit a polynomial kernel for layered DAGs of maximum finite distance~$2$, %
unless the polynomial hierarchy collapses to its third level. %} 
Moreover, we show that both \BD{} and \MP{} are polynomial-time solvable for layered DAGs with a single source. %
We also show that both problems are FPT when parameterized by the solution cost/solution size together with the maximum out-degree, and as well, by the vertex cover number. Moreover it follows from a powerful meta-theorem of~\cite{FOnowheredense} that \BD{} is FPT when parameterized by solution cost, on inputs whose underlying graphs belong to a nowhere dense class.

The resulting complexity landscape is represented in Fig.~\ref{fig:summary2}. %
We start with some definitions in Section~\ref{sec:prelim}. We prove our results for \BD{} in Section~\ref{sec:BD}. The results for \MP{} are presented in Section~\ref{sec:MP}. We conclude in Section~\ref{sec:conclu}. 
\definecolor{darkred}{RGB}{170,10,50}
\definecolor{normalred}{RGB}{220,0,0}
\definecolor{grassgreen}{RGB}{100,160,20}
\definecolor{darkgreen}{RGB}{0,150,20}
\definecolor{pacificorange}{RGB}{220,100,0}

%NEW FIGURE WITH SAME CLASSES

\begin{figure}[!hbt]
\centering
\scalebox{0.65}{\begin{tikzpicture}[node distance=10mm]
    \begin{scope}
      \node[fptbox] (all) at (0,0) {all digraphs};
      \node[fptbox] (simple) [below = of all, xshift=12mm,yshift=5mm] {no $2$-cycles} edge[fptedge] (all);
      \node[fptbox] (bip) [left = of simple, xshift=6mm,yshift=0mm] {bipartite} edge[fptedge] (all);
      \node[fptbox] (simplebip) [below = of all, xshift=-5mm,yshift=-8mm] {bipartite, no $2$-cycles} edge[fptedge] (simple) edge[fptedge] (bip);      
      
      \node[fptbox] (DAG) [below = of simple, xshift=10mm,yshift=-11mm] {DAGs}; 
      \draw[fptedge] (DAG) .. controls +(0,1) and +(1,-1) .. (simple);
      
      \node[fptbox] (nd) [below left = of simplebip, xshift=15mm,yshift=-1mm] {\tworows{nowhere}{dense}};
      \draw[fptedge] (nd) .. controls +(-0.5,2.5) and +(-2.5,-0.5) .. (all);
      
      \node[fptbox] (maxoutdeg) [below left = of nd, xshift=15mm,yshift=-5mm] {\tworows{bounded max.}{out-deg.}};
      \draw[fptedge] (maxoutdeg) .. controls +(0,3) and +(-4,-0.5) .. (all);
      
      \node[fptbox] (lay) [right = of maxoutdeg, xshift=5mm,yshift=0mm] {layered DAGs} edge[fptedge,bend right=10] (DAG) edge[fptedge,bend left=10] (simplebip);

      \node[fptbox] (maxdeg) [below = of maxoutdeg, xshift=0mm,yshift=2mm] {\tworows{bounded}{max. deg.}} edge[fptedge] (maxoutdeg);
      \draw[fptedge] (maxdeg) .. controls +(1.75,1) and +(0,-3) .. (nd);
      
      \node[fptbox] (plan) [below = of maxdeg, xshift=5mm,yshift=-5mm] {\tworows{planar layered DAGs}{of small max. deg.}} edge[fptedge,bend left=10] (maxdeg);
      \draw[fptedge] (plan) .. controls +(0.5,1.5) and +(-2,-0.7) .. (lay);
      
      \node[fptbox] (ss) [right = of maxdeg, xshift=21mm,yshift=0mm] {\tworows{single-sourced}{DAGs}};
      \draw[fptedge] (ss) .. controls +(1,2) and +(0.5,-1) .. (DAG);   
      
      \node[fptbox] (ssmd) [right = of plan, xshift=2mm,yshift=0mm] {\tworows{single-sourced DAGs}{of small max. deg.}} edge[fptedge] (ss)  edge[fptedge,bend left=15] (maxdeg);

      \node[fptbox] (ssl) [below = of plan, xshift=6mm,yshift=3mm] {\tworows{single-sourced}{layered DAGs}};
      \draw[fptedge] (ssl) .. controls +(2,1.2) and +(-1.5,-2) .. (lay); 
      \draw[fptedge] (ssl) .. controls +(2.5,1) and +(-2.5,-1.33) .. (ss); 
      
      \node[fptbox] (vc) [right = of ssl, xshift=0mm,yshift=0mm] {\tworows{bounded}{vertex cover}};
      \draw[fptedge] (vc) .. controls +(-2.75,2) and +(0.5,-3) .. (nd); 

      \node[fptbox] (strong) [right = of ss, xshift=3mm,yshift=0mm] {\tworows{strongly}{connected}};
      \draw[fptedge] (strong) .. controls +(-2,3) and +(4,-0.5) .. (all);

      \node[fptbox] (diam) [right = of vc, xshift=-3mm,yshift=0mm] {\tworows{bounded}{diameter}};
      \draw[fptedge] (diam) .. controls +(0.25,3) and +(-0.75,-1) .. (strong);
      
      \node[fptbox] (sym) [right = of ssmd, xshift=0mm,yshift=0mm] {symmetric};
      \draw[fptedge] (sym) -- (strong);

      \node (caption) at ($(ssl.south)+(3,-1)$) {\large{(a) Complexity landscape of \BD}};

      \node (wh) at ($(sym.east |- simplebip.south)+(-0.8,0)$) {W[2]-hard};
      \node[above = of wh, yshift=-12mm] {$\uparrow$};
      \draw[color=darkred,line width=2pt,-,rounded corners] ($(maxoutdeg.west |- simplebip.south)+(0,-0.3)$) -- ($(sym.east |- simplebip.south)+(0,-0.3)$);

      \draw[color=darkgreen,line width=2pt,-,rounded corners] ($(maxoutdeg.west |- simplebip.south)+(0,-0.5)$) -- ($(sym.east |- simplebip.south)+(0,-0.5)$);
      \node (fpt) at ($(maxoutdeg.west |- simplebip.south)+(0.4,-0.8)$) {FPT};
      \node[below = of fpt, yshift=11.5mm] {$\downarrow$};  

      \draw[color=normalred,line width=2pt,-,rounded corners] ($(maxoutdeg.west |- plan.south)+(0,-0.2)$) -- ($(plan.south east)+(0,-0.2)$) .. controls +(0.3,0.2) and +(-0.3,-0.2) .. ($(ss.north west)+(0,0.2)$) -- ($(strong.north east)+(0,0.2)$);
      \node (nph) at ($(strong.north east)+(-0.7,0.5)$) {NP-hard};
      \node[above = of nph, yshift=-11.5mm] {$\uparrow$};  

      \draw[color=grassgreen,line width=2pt,-,rounded corners] ($(maxoutdeg.west |- plan.south)+(0,-0.4)$) -- ($(ssmd.east |- plan.south)+(0,-0.4)$) .. controls +(0.2,0.2) and +(-0.2,-0.2) .. ($(sym.north west)+(0,0.2)$) -- ($(sym.north east)+(0,0.2)$);
      \node (poly) at ($(plan.south west)+(0.4,-0.7)$) {poly};
      \node[below = of poly, yshift=12mm] {$\downarrow$};  
    \end{scope}
   
    \begin{scope}[xshift=12.5cm,yshift=0mm]
      \node[fptbox] (all) at (0,0) {all digraphs};
      \node[fptbox] (simple) [below = of all, xshift=12mm,yshift=5mm] {no $2$-cycles} edge[fptedge] (all);
      \node[fptbox] (bip) [left = of simple, xshift=6mm,yshift=0mm] {bipartite} edge[fptedge] (all);
      \node[fptbox] (simplebip) [below = of all, xshift=-5mm,yshift=-8mm] {bipartite, no $2$-cycles} edge[fptedge] (simple) edge[fptedge] (bip);      
      
      \node[fptbox] (DAG) [below = of simple, xshift=10mm,yshift=-11mm] {DAGs}; 
      \draw[fptedge] (DAG) .. controls +(0,1) and +(1,-1) .. (simple);
      
      \node[fptbox] (nd) [below left = of simplebip, xshift=15mm,yshift=-1mm] {\tworows{nowhere}{dense}};
      \draw[fptedge] (nd) .. controls +(-0.5,2.5) and +(-2.5,-0.5) .. (all);
      
      \node[fptbox] (maxoutdeg) [below left = of nd, xshift=15mm,yshift=-5mm] {\tworows{bounded max.}{out-deg.}};
      \draw[fptedge] (maxoutdeg) .. controls +(0,3) and +(-4,-0.5) .. (all);
      
      \node[fptbox] (lay) [right = of maxoutdeg, xshift=5mm,yshift=0mm] {layered DAGs} edge[fptedge,bend right=10] (DAG) edge[fptedge,bend left=10] (simplebip);

      \node[fptbox] (maxdeg) [below = of maxoutdeg, xshift=0mm,yshift=2mm] {\tworows{bounded}{max. deg.}} edge[fptedge] (maxoutdeg);
      \draw[fptedge] (maxdeg) .. controls +(1.75,1) and +(0,-3) .. (nd);
      
      \node[fptbox] (plan) [below = of maxdeg, xshift=5mm,yshift=-5mm] {\tworows{planar layered DAGs}{of small max. deg.}} edge[fptedge,bend left=10] (maxdeg);
      \draw[fptedge] (plan) .. controls +(0.5,1.5) and +(-2,-0.7) .. (lay);
      
      \node[fptbox] (ss) [right = of maxdeg, xshift=21mm,yshift=0mm] {\tworows{single-sourced}{DAGs}};
      \draw[fptedge] (ss) .. controls +(1,2) and +(0.5,-1) .. (DAG);   
      
      \node[fptbox] (ssmd) [right = of plan, xshift=2mm,yshift=0mm] {\tworows{single-sourced DAGs}{of small max. deg.}} edge[fptedge] (ss)  edge[fptedge,bend left=15] (maxdeg);

      \node[fptbox] (ssl) [below = of plan, xshift=6mm,yshift=3mm] {\tworows{single-sourced}{layered DAGs}};
      \draw[fptedge] (ssl) .. controls +(2,1.2) and +(-1.5,-2) .. (lay); 
      \draw[fptedge] (ssl) .. controls +(2.5,1) and +(-2.5,-1.33) .. (ss); 
      
      \node[fptbox] (vc) [right = of ssl, xshift=0mm,yshift=0mm] {\tworows{bounded}{vertex cover}};
      \draw[fptedge] (vc) .. controls +(-2.75,2) and +(0.5,-3) .. (nd); 

      \node[fptbox] (strong) [right = of ss, xshift=3mm,yshift=0mm] {\tworows{strongly}{connected}};
      \draw[fptedge] (strong) .. controls +(-2,3) and +(4,-0.5) .. (all);

      \node[fptbox] (diam) [right = of vc, xshift=-3mm,yshift=0mm] {\tworows{bounded}{diameter}};
      \draw[fptedge] (diam) .. controls +(0.25,3) and +(-0.75,-1) .. (strong);
      
      \node[fptbox] (sym) [right = of ssmd, xshift=0mm,yshift=0mm] {symmetric};
      \draw[fptedge] (sym) -- (strong);

      \node (caption) at ($(ssl.south)+(3,-1)$) {\large{(b) Complexity landscape of \MP}};

      \node (wh) at ($(sym.east |- all.south)+(-0.8,0)$) {W[1]-hard};
      \node[above = of wh, yshift=-12mm] {$\uparrow$};
      \draw[color=darkred,line width=2pt,-,rounded corners] ($(maxoutdeg.west |- all.south)+(0,-0.3)$) -- ($(sym.east |- all.south)+(0,-0.3)$);

      \draw[color=darkgreen,line width=2pt,-,rounded corners] ($(maxoutdeg.north west)+(0,1)$) -- ++(2.5,0) .. controls +(0.3,-0.2) and +(-0.3,0.2) .. ($(ss.west |- ssmd.north)+(0.1,0.2)$) -- ($(ssmd.north east)+(0,0.2)$) .. controls +(0.1,-0.2) and +(-0.2,0.2) .. ($(sym.south west)+(0,-0.2)$) -- ($(sym.south east)+(0,-0.2)$);
      \node (fpt) at ($(maxoutdeg.north west)+(0.4,0.7)$) {FPT};
      \node[below = of fpt, yshift=11.5mm] {$\downarrow$};

      \draw[color=normalred,line width=2pt,-,rounded corners] ($(maxoutdeg.west |- plan.south)+(0,-0.2)$) -- ($(ssmd.east |- plan.south)+(0,-0.2)$) .. controls +(0.2,0.2) and +(-0.2,-0.2) .. ($(strong.north west)+(0,0.2)$) -- ($(strong.north east)+(0,0.2)$);
      \node (nph) at ($(strong.north east)+(-0.7,0.5)$) {NP-hard};
      \node[above = of nph, yshift=-11.5mm] {$\uparrow$};  

      \draw[color=grassgreen,line width=2pt,-,rounded corners] ($(maxoutdeg.west |- plan.south)+(0,-0.4)$) -- ($(sym.east |- plan.south)+(0,-0.4)$);
      \node (poly) at ($(plan.south west)+(0.4,-0.7)$) {poly};
      \node[below = of poly, yshift=12mm] {$\downarrow$};  
    \end{scope}
\end{tikzpicture}}
\caption{Complexity landscape of \BD{} and \MP{} for some classes of digraphs (all considered digraphs are assumed to be weakly connected). An arc from class $A$ to class $B$ indicates that $A$ is a subset of $B$. Parameterized complexity results are for parameter solution cost/solution size.}\label{fig:summary2}
\end{figure}
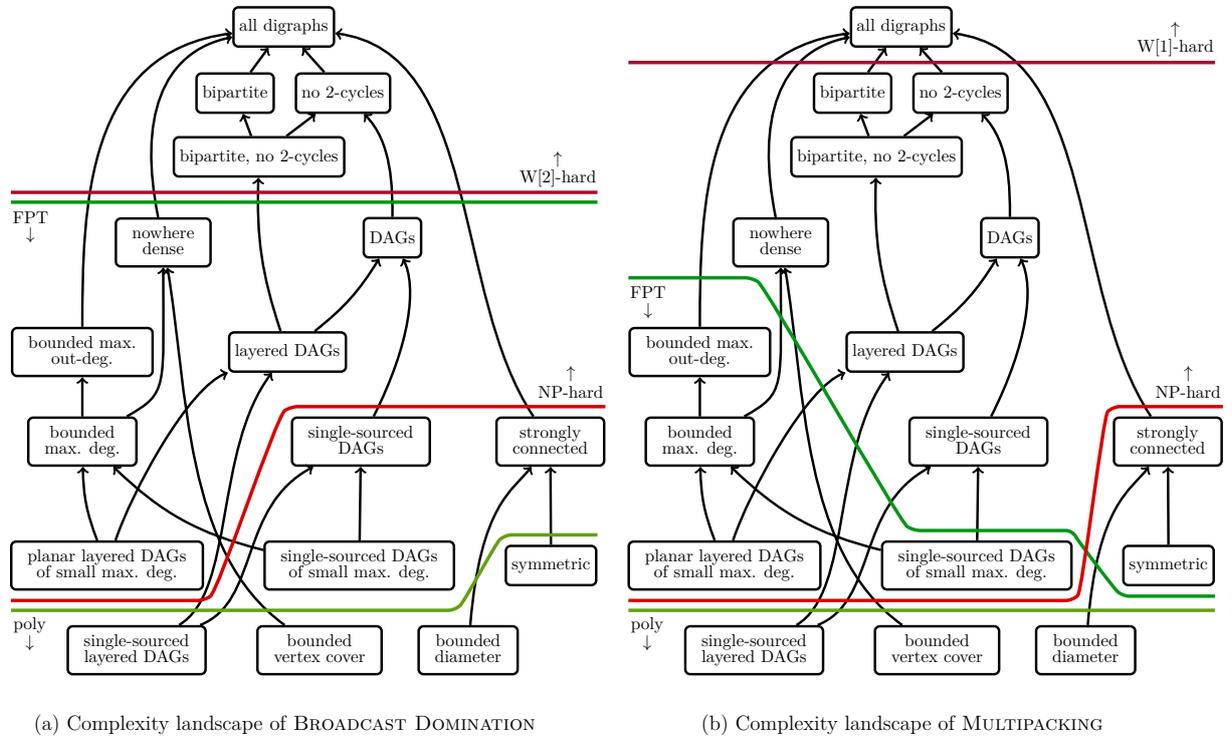

\section{Preliminaries}\label{sec:prelim}
\setcounter{footnote}{0}

\textbf{Directed graphs.} We mainly consider digraphs, usually denoted $D = (V,A)$\footnote{Our reductions will also use undirected graphs, denoted $G = (V,E)$ with $V = \{v_1, \ldots, v_n\}$ and $E = \{e_1, \ldots, e_m\}$.}, where $V$ is the set of vertices and $A$ the set of arcs. 
For an arc $uv \in A$, we say that $v$ is an \emph{out-neighbor} of $u$, and $u$ an \emph{in-neighbor} of $v$.
Given a subset of vertices $V' \subseteq V$, we define the digraph \emph{induced} by $V'$ as 
$D' = (V', A')$ where $A' = \{uv \in A:\ u \in V'\ \textstyle{and}\ v \in V'\}$. We denote such an induced 
subdigraph by $D[V']$. 
A directed path from a vertex $p_1$ to $p_l$ is a sequence $\{p_1, \ldots, p_l\}$ such that $p_i\in V$ and 
$p_ip_{i+1} \in A$ for every $1 \leqslant i < l$. When $p_1 = p_l$, it is a directed cycle. 
A digraph is \emph{acyclic} whenever it does not contain any directed cycle as an induced subgraph. An acyclic digraph is called a \emph{DAG} for short.
The \emph{(open) out-neighborhood} of a vertex $v \in V$ is the set $N^+(v) = \{u \in V:\ vu \in A\}$, % 
and its \emph{closed out-neighborhood} is $N^+[v] = N^+(v) \cup \{v\}$. We define similarly the open 
and closed in-neighborhoods of $v$ and denote them by $N^-(v)$ and $N^-[v]$, respectively. %
A \emph{source} is a vertex $v$ such that $N^-(v) = \emptyset$. 
For the sake of readability, we always mean out-neighborhood when speaking of the \emph{neighborhood} 
of a vertex. 
A DAG $D = (V,A)$ is \emph{layered} when its vertex set can be partitioned into $\{V_0, \ldots, V_t\}$ such that $N^-(V_0) = \emptyset$ and $N^+(V_t) = \emptyset$ (vertices of $V_0$ and 
$V_t$ are respectively called \emph{sources} and \emph{sinks}), and $uv \in A$ 
implies that $u \in V_i$ and $v \in V_{i+1}$, $0 \leqslant i < t$. 
A \emph{single-sourced layered DAG} is a layered DAG with only one source, that is, satisfying $\size{V_0} = 1$. %
A digraph is \emph{bipartite} or \emph{planar} if its underlying undirected graph has the corresponding property. Every layered digraph is bipartite. %
Given two vertices $u$ and $v$, we denote by $d(u,v)$ the length of a shortest directed path from $u$ to $v$. For a vertex $v \in V$ and an integer $d$, we define the \emph{ball of radius $d$ centered at $v$} by $B^+_d(v) = \{u \in V:\ d(v,u) \leqslant d\} \cup \{v\}$. %
The \emph{eccentricity} of a vertex $v$ in a digraph $D$ is the largest (finite) distance between $v$ and any vertex of $D$, denoted $\ecc(v) := \max_{u \in V} d(v,u)$. %
A digraph is \emph{strongly connected} if for any two vertices $u,v$, there is a directed path from $u$ to $v$, and \emph{weakly connected} if its underlying undirected graph is connected. We will assume that all digraphs considered here are weakly connected (if not, each component can be considered independently). The \emph{diameter} is the maximum directed distance $\max_{u,v \in V} d(u,v)$ between any two vertices $u$ and $v$ of $G$. If the digraph is  not strongly connected, then the diameter is infinite. The \emph{maximum finite distance} of a digraph $D$ is the largest finite directed distance between any two vertices of $G$, denoted $\mfd(D) := \max_{u,v \in V, d(u,v)<\infty} d(u,v)$. % Note that when $\mfd(D)=1$, $D$ has only sources and sinks, and \BD{} and \MP{} are easy to solve for $D$.
Consider a dominating broadcast $f : V(D) \to \mathbb{N}$ on $D$. % 
The set of \emph{broadcast dominators} is defined as $V_f = \{v \in V:\ f(v) > 0\}$. 
For any set $S \subseteq V$ of vertices of $D$, we define $f(S)$ as the value $f(S) = \sum_{u\in S} f(u)$. 

\medskip

\noindent \textbf{Parameterized complexity.} A \emph{parameterized problem} is a decision problem together with a \emph{parameter}, that is, an integer $k$ depending on the instance.
A problem is \emph{fixed-parameter tractable} (FPT for short) if it can be solved in time $f(k)\cdot|I|^c$ for an instance $I$ of size $|I|$ with parameter $k$, where $f$ is a computable function and $c$ is a constant. %
Given a parameterized problem $P$, a \emph{kernel} is a function which associates to each instance of $P$ an equivalent instance of $P$ whose size is bounded by a function $h$ of the parameter. When $h$ is a polynomial, the kernel is said to be \emph{polynomial}.
An \emph{\FPTR{}} between two parameterized problems $P$ and $Q$ is a function mapping an instance $(I,k)$ of $P$ to an instance $(f(I), g(k))$ of $Q$, where $f$ and $g$ are computable in FPT time with respect to parameter $k$, and where $I$ is a YES-instance of $P$ if and only if $f(I)$ is a YES-instance of $Q$.  %
When moreover $f$ can be computed in polynomial time and $g$ is polynomial in $k$, we say that the reduction is a \emph{polynomial time and parameter transformation{}}~\cite{BTY11}. %
Both reductions can be used to derive conditional lower 
bounds: if a parameterized problem $P$ does not admit an FPT algorithm (resp. a polynomial kernel) and there exists 
an \FPTR{} (resp. a \PTPT{}) from $P$ to a parameterized problem $Q$, 
then $Q$ is unlikely to admit an FPT algorithm (resp. a polynomial kernel). Both implications rely 
on certain standard complexity hypotheses; we refer the reader to the book~\cite{CFK+15} for details. 

\section{Complexity of \BD{}}\label{sec:BD}

\subsection{Hardness results}

\begin{theorem}
\label{thm:dbds:npc}
\BDS{} is NP-complete, even for planar layered DAGs of maximum degree~$4$ and maximum finite distance~$2$.
\end{theorem}

\begin{proof}
We will reduce from \textsc{Exact Cover by $3$-Sets}, defined as follows.

\Pb{\textsc{Exact Cover by $3$-Sets}} 
{A set $X$ of $3k$ elements (for some $k \in \mathbb{N}$), and a set $\mathcal{T} = \{t_1, \ldots, t_n\}$ of triples from $X$.}
{a subset $\mathcal{S}$ of $k$ triples from $\mathcal{T}$ such that each element of $X$ appears in (exactly) one triple in $\mathcal{S}$}
{Question}

\textsc{Exact Cover by $3$-Sets} is NP-hard even when the incidence bipartite graph of the input is planar and each element appears in at most three triples~\cite{DF86}. We will reduce any such instance $(X,\mathcal{T})$ of \textsc{Exact Cover by $3$-Sets} to an instance $(D=(V',A'),k')$ of \BDS{}. 

\medskip

 We create $V'$ by taking two copies $T^1$, $T^2$ of $\mathcal{T}$ and one copy of $X$.
More precisely, we let $T^i = \{t_j^i:\ 1 \leqslant i \leqslant n$ for $i \in \{1,2\}$.
We now add an arc from a vertex $t^1_i \in T^1$ to its corresponding vertex $t^2_i$ in $T^2$, and from a vertex $t_i^2 \in T^2$ to all elements of $X$ that are contained in $t_i$ in $(X,\mathcal{T})$.
See also Figure~\ref{fig:dbds:npc}. Formally:
    $$A' = \{t_i^1t_i^2:\ 1 \leqslant i \leqslant n \} \bigcup\ \{t_i^2x:\ x \in T_i,\ 1 \leqslant i \leqslant n\}$$
The construction can be done in polynomial time, and there is no cycle in $D$: arcs go either from $T^1$ to $T^2$ or from $T^2$ to $X$. Hence $D$ is a layered DAG with three layers and thus, maximum finite distance~$2$.
In fact $D$ is obtained from the bipartite incidence graph of $(X,\mathcal{T})$ (which is planar and of maximum degree~$3$) reproduced on the vertices of $T^2\cup X$, by adding pendant vertices (those from $T^1$) to those of $T^2$, orienting the arcs as required.
Thus, the maximum degree of $D$ is~$4$ and $D$ is planar.

\begin{figure}[ht!]
    \centerline{\includegraphics[scale=1.5]{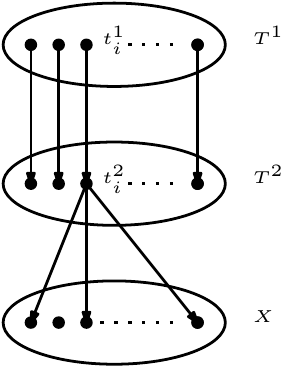}}
    \caption{Sketch of the DAG built in the construction of the proof of Theorem~\ref{thm:dbds:npc}.}\label{fig:dbds:npc}
\end{figure}

\begin{claim}
The instance $(X,\mathcal{T})$ is a YES-instance if and only if the digraph $D$ has a dominating broadcast of size $k'=n+k$. 
\end{claim}
%\begin{proofclaim} 

$\Rightarrow$ Given a solution $\mathcal{S}$ of $(X,\mathcal{T})$, set $f(t_i^1) = 2$ for all $t_i \in \mathcal{S}$, $f(t_i^1) = 1$ for each of the $n-k$ remaining vertices of $T^1$ and $f(v)=0$ for all vertices
of $T^2$ and $X$. %Every vertex of $V_1$ are trivially dominated by $f$. 
For every vertex $t^2_i \in T^2,$ we have $d(t^1_i,t^2_i) = 1$.
Similarly, for every vertex $x \in X,\ d(t^1_i,x) \leq 2$ holds for the vertex $t^1_i$
such that $t_i$ is in $\mathcal{S}$ and contains $x$ in $(X,\mathcal{T})$. Since every vertex $t_i^1$ of $T^1$ satisfies $f(t_i^1) \geqslant 1$, it is covered by itself, and
it follows that $f$ is a dominating broadcast of size $n+k$. \\
$\Leftarrow$ Let us now consider the case where we are given a dominating broadcast for $D$ of cost $n+k$.
Note that since the maximum finite distance is~$2$, we can assume $f : V' \rightarrow \{0,1,2\}$.
Remark that the vertices of $T^1$ are $n$ sources. Therefore, any broadcast needs
to set $f(t_i^1) \geq 1$ for each $t_i^1 \in T^1$, and this covers all vertices of $T^1$ and $T^2$.
It remains to cover vertices of $X$ with a cost of $k$,
which can be done by setting $f(t_i^1) = 2$ for some vertices of $T^1$ and $f(t_j^2) = 1$ for some vertices
of $T^2$. Notice that it is never useful to set $f(x)=1$ for some vertex $x \in X$: setting
an additional cost of $1$ to any $f(t_i^2)$ such that $t_i^2 \in A'$ is always better.
Hence, the corresponding set of triples is a valid cover of $(X, \mathcal{T})$. (And it is an exact cover because there are $3k$ elements covered by $k$ triples.)
%which can be done only if there is a set of $k$ triples $T_i$ of $\mathcal{T}$ with $f(t_i^1)=2$ or $f(t_i^2)=1$. This set is a valid solution of $(X,\mathcal{T})$.
%\end{proofclaim}
\end{proof}

We next give two parameterized reductions for \BDS{}.

\begin{theorem}
\label{thm:dbds:w}
\BDS{} parameterized by solution cost $k$ is W[2]-hard, even on digraphs of maximum finite distance~$2$, and on bipartite digraphs without directed~$2$-cycles of maximum finite distance~$6$.
\end{theorem}

\begin{proof} 
We provide two reductions from the W[2]-hard \MDS{} problem~\cite{Casel18}, defined as follows. 
%Our reduction gives bipartite instances\fullversion{of diameter~$6$} with no directed $2$-cycles.

\Pb{\MDS{}}
{A graph $G=(V,E)$ with $V$ partitioned into $k$ sets 
$\{V_1, \ldots, V_k\}$, for an integer $k \in \mathbb{N}$.}
{A dominating set $S$ of $G$ such that $|S \cap V_i |=1$ for every $1 \leqslant i \leqslant k$.}
{Output}

We first provide a reduction that gives digraphs with directed $2$-cycles. 

\paragraph{Construction 1.} We build an instance $(D = (V',A'), k')$ of \BDS{} as follows.
To obtain the vertex set $V'$, we duplicate $V$ into two sets $V^1$ and $V^2$. 
Following the partition of $V$ into $k$ sets, we let $V^1 = \{V^1_1, \ldots, V^1_k\}$ and $V^2 = \{V_1^2, \ldots, V_k^2\}$. 
We then add every possible arc within $V^1_i$ ($1 \leqslant i \leqslant k$), and an arc from a vertex $v$ in $V^1$ to each vertex of $V^2$ corresponding to a vertex from the closed neighborhood of $v$ in $G$. 
Altogether, $V' = V^1 \cup V^2$. Finally, we set $k'=k$. See Figure~\ref{fig:dbds:w} for an illustration. %\todo{inverser indices et exposants dans cette figure} 
Clearly $\mfd(D)=2$.

\begin{figure}[ht!]
    \centerline{\includegraphics[scale=1.5]{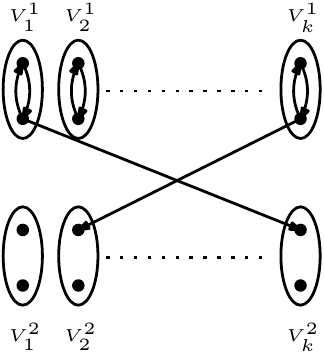}}
    \caption{Sketch of the built digraph $D$ in the first reduction of the proof of Theorem~\ref{thm:dbds:w}. }\label{fig:dbds:w}
\end{figure}

\begin{claim}
    The graph $G$ has a multicolored dominating set of size $k$ if and only if the digraph $D$ has a dominating broadcast of cost $k$.
\end{claim}

\begin{proofclaim} $\Rightarrow$ Let $S \subseteq V$ be a multicolored dominating set of size $k$ of $G$. 
    We claim that setting $f(v) = 1$ for every vertex $v$ of $V^1$
    such that the corresponding vertex $v$ of $G$ is in $S$, yields a dominating broadcast of cost $k$. To see this, notice that each vertex $v\in V^1_i$ ($1 \leqslant i \leqslant k$)
    with cost $1$ covers $V^1_i$. Now, since these vertices of cost~$1$ form a
    dominating set in $G$, they cover  the vertices of $V^2$ corresponding to their closed neighborhood in $G$, 
    and hence $f$ is a dominating broadcast.\\
    $\Leftarrow$ Assume now that $D$ has a dominating broadcast $f$ of cost $k$. Notice first that 
    any set $V^1_i$ ($1 \leqslant i \leqslant k$) must contain a vertex $v$ such that $f(v) 
    \geqslant 1$. Since $f$ has cost $k$, this means that for every vertex $w\in V^2$, $f(w) = 0$. 
    It follows that one needs to cover the vertices of $V^2$ using $k$ vertices in $V^1$, 
    which can be done only if there is a multicolored dominating set of size~$k$ in $G$. 
\end{proofclaim}

We now give a similar but more involved construction, which gives bipartite instances of maximum finite distance~$6$ and no directed $2$-cycles.

\paragraph{Construction 2.} 
We build an instance $(D' = (V',A'), k')$ of \BDS{} as follows.
To obtain the vertex set $V'$, we multiplicate $V$ into four sets $V^0$, $V^1$, $V^2$ 
and $V^3$ and we will have a set $M$ of subdivided vertices. The set $V^0\cup V^1$ will induce an 
oriented complete bipartite graph, while $V^2\cup V^3$ will induce a matching. 
Following the partition of $V$ into $k$ sets, for $0\leqslant i \leqslant 3$, we let $V^i = 
\{V_1^i, \ldots, V_k^i\}$. For a vertex $v\in V$, for $0\leqslant i \leqslant 3$ its 
copy in $V^i$ is denoted $v^i$. We assume that $\size{V_i} \geqslant 2$, since otherwise one must take the only vertex in $V_i$. 
For each $1\leq i\leq k$ we then add the following arcs:
\begin{itemize}
\item for every pair $v,w$ of distinct vertices of $V_i$, we add an arc from $v^0$ to $w^1$; 
\item for every $v\in V_i$, we add an arc from $v^1$ to $v^0$;
\item for every $v\in V_i$, we add an arc from $v^2$ to $v^3$.
\end{itemize}

\begin{figure}[ht!]
    \centerline{\includegraphics[scale=1.5]{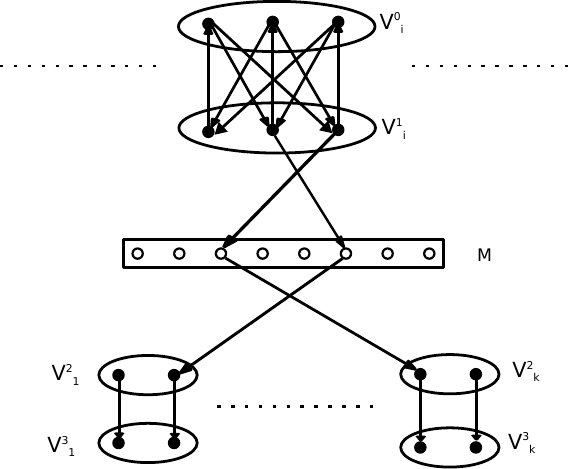}}
    \caption{Sketch of the built digraph $D''$ in the second reduction of the proof of Theorem~\ref{thm:dbds:w}. }\label{fig:dbds:w-2}
\end{figure}

Moreover, for every edge $vw$ in $G$, we add an arc from $v^1$ to $w^2$, and we subdivide it once. 
The set of all subdivision vertices is called $M$. Finally, we set $k'=3k$. %See Figure~\ref{fig:dbds:w-2} for an illustration. 
\fullversion{It is clear that $\mfd(D')=6$ (shortest paths of length~$6$ exist from vertices of $V^0$ 
to vertices of $V^3$, but no longer shortest paths exist). }
The digraph has clearly no directed 
$2$-cycles, and is bipartite with sets $V^0\cup M\cup V^3$ and $V^1\cup V^2$. %

\begin{claim}
    The graph $G$ has a multicolored dominating set of size $k$ if and only if the digraph $D'$ has a dominating broadcast of cost $3k$.
\end{claim}

\begin{proofclaim}$\Rightarrow$ Let $S \subseteq V$ be a multicolored dominating set of size $k$ of $G$. 
    We claim that setting $f(v^1) = 3$ for every vertex $v^1$ of $V^1$ such that $v \in S$ yields a 
    dominating broadcast of cost $3k$. To see this, notice first that each such vertex 
    belonging to $V_i^1$, $1 \leqslant i \leqslant k$, covers the whole set $V_i^0\cup V_i^1$ and all the vertices of 
    $M$ with an in-neighbor in $V_i^1$. Now, each vertex $v^1$ with $v\in S$ covers (at distance~$3$) 
    each vertex $w^2$ and $w^3$ of $V^2\cup V^3$ such that $w$ is in the closed neighborhood of $v$ in 
    $G$. Since $S$ is dominating, $f$ is thus a dominating broadcast.  \\
    $\Leftarrow$ Assume now that $D'$ has a dominating broadcast $f$ of size $3k$. First, we claim that for 
    every $i$ with $1 \leqslant i \leqslant k$, we need a total cost of~$3$ for the vertices in $V_i^0\cup 
    V_i^1$. Indeed, for a vertex $v\in V_i$, if $f(v^0)=2$, $v^0$ does not cover $V^1$. If $f(v^1)=2$, 
    no vertex $w^0$ with $w\neq v$ and $w\in V_i$ is covered. Clearly, we cannot cover the vertices of 
    $V_i^0\cup V_i^1$ with two vertices broadcasting at cost~$1$. Thus, we can assume that there is is a 
    total cost of exactly~$3$ on the vertices of $V_i^0\cup V_i^1$ for $1 \leqslant i \leqslant k$, and 
    each vertex $v$ of $V^2\cup V^3\cup M$ satisfies $f(v)=0$. 
    We now prove that there exists a vertex $v$ of $V_i^0 \cup V_i^1$, $1 \leqslant i \leqslant k$ such 
    that $f(v) = 3$. 
    First, since a vertex $v^1$ of $V_i^1$ with $f(v^1)=2$ does not cover the vertices of $V_i^0$ 
    (except for $v^0$), it is not possible to cover $V_i^0 \cup V_i^1$ with a cost of~$1$ on another vertex. 
    Similarly, since a vertex $v^0$ of $V_i^0$ with $f(v^0)=2$ does not cover $v^1$, an 
    additional cost of~$1$ cannot cover $v^1$ and all vertices of $M$ that are 
    out-neighbors of vertices in $V_i^1$. Similarly, we cannot have three vertices with a broadcasting cost 
    of~$1$ each. Thus, there is a vertex of $V_i^0\cup V_i^1$ with a broadcast cost of~$3$. Notice that 
    it cannot be a vertex of $V_i^0$, since otherwise the out-neighbors of $V^1_i$ in $M$ are not 
    covered. Thus there is a vertex $v^1$ in $V_i^1$ with $f(v^1)=3$. This covers, in particular, all 
    the vertices $w^2,w^3$ of $V_i^2\cup V_i^3$ such that $vw$ is an edge in $G$, and no other vertex of 
    $V_i^2\cup V_i^3$. It follows that the set of vertices $v$ of $V$ such that $f(v^1)=3$ forms a 
    dominating set of $G$ of size $k$.
\end{proofclaim}
Thus, the proof is complete.
 \end{proof}

\subsection{Complexity and algorithms for (layered) DAGs}\label{sec:BD-DAG}

We now address the special cases of (layered) DAGs. Note that \textsc{Dominating Set} remains W[2]-hard on DAGs by a reduction from~\cite[Theorem 6.11.2]{OK14}. %
In contrast, we now give an FPT algorithm for \BD{} on DAGs that counterbalances the W[2]-hardness result.

\begin{theorem}%
\label{thm:dbds:fpt}
\BDS{} parameterized by solution cost $k$ can be solved in FPT time $2^{O(k\log k)}n^{O(1)}$ time for DAGs of order $n$.
\end{theorem}

The proof relies on the following proposition, which is reminiscent of a stronger statement of Dunbar et al.~\cite{DEH+06} for undirected graphs (stating that there always exists an optimal dominating broadcast where each vertex is covered exactly once, which is false for digraphs).

\begin{proposition}%
\label{prop:dbds:eff}
    For any digraph $D = (V,A)$, there exists an \emph{optimal} dominating broadcast such that 
    every broadcast dominator is covered by itself only. 
\end{proposition}

\begin{proof} Let $f$ be an optimal dominating broadcast of $D$, and assume there exists two vertices $u,v \in V$ 
    such that $f(v) \geqslant 1$ and $f(u) \geqslant d(u,v)$. 
    In this case, $v$ is covered by both $u$ and itself. 
    Notice that %the maximum finite distance in $D$ is at least 
    $d(u,v) + f(v) > f(u)$, %the latter inequality holding 
    since otherwise setting $f(v)$ to $0$ 
    would result in a better dominating broadcast. We claim that setting $f(u)$ to $d(u,v) + f(v)$ 
    and $f(v)$ to $0$ yields an optimal dominating broadcast $f_u$. 
    Notice that since $d(u,v) + f(v) > f(u)$, any vertex covered by $u$ in $f$ is still 
    covered in $f_u$. Similarly, 
    any vertex covered by $v$ in $f$ is now covered by $u$ in $f_u$. Finally, we have 
    $f(u)+f(v) \geqslant f_u(u) + f_u(v)$ since $f_u(u) = d(u,v) + f(v) \leqslant f(u) + f(v)$ and 
    $f_u(v) = 0$, implying that the cost of $f_u$ is at most the cost of $f$. 
\end{proof}

We can now prove Theorem~\ref{thm:dbds:fpt}.

\begin{proof}[Proof of Theorem~\ref{thm:dbds:fpt}] 
Let $D = (V,A)$ be a DAG. We consider the set $V_0$ of sources of $D$.
Observe that for every $s \in V_0,\ f(s) \geqslant 1$ must hold. 
In particular, this means that $|V_0| \leqslant k$ (otherwise we return NO). %
We provide a branching algorithm based on this simple observation and on Proposition~\ref{prop:dbds:eff}. We start with an initial broadcast $f$ consisting of setting $f(s) = 1$ for every vertex $s$ in $V_0$. 
At each step of the branching algorithm, we let $N_f = \cup_{v \in V_f} B^+_{f(v)}(v)$ be the set of currently covered vertices, and we consider the digraph $D_f = D[V \setminus N_f]$. 
Notice that $D_f$ is acyclic and hence contains a source $u$. Since every vertex of $N_f \setminus V_f$ is covered, we may assume by 
Proposition~\ref{prop:dbds:eff} that in the sought optimal solution, 
$u$ is only covered by itself or by a vertex in $V_f$. 
This means that one needs to branch on at most $k+1$ 
distinct cases: either setting $f(u) = 1$, or increasing the cost of one of its at most $k$ broadcasting ancestors in $V_f$. 
At every branching, the parameter $k$ decreases by~$1$, which ultimately gives an $O^*(2^{k \log{k}})$-time algorithm and completes the proof of Theorem~\ref{thm:dbds:fpt}. 
 \end{proof}

We will now complement the previous result by a negative one, which can be proved using a reduction similar to the one in Theorem~\ref{thm:dbds:npc} but from \textsc{Hitting Set}, defined as follows. %

\Pb{\textsc{Hitting Set}}
{A \emph{universe} $U$ of elements, a collection $\mathcal F$ of subsets of $U$, an integer $k \in \mathbb{N}$.}
{a \emph{hitting set} $S$ of size $k$, that is, a set of $k$ elements from $U$ such that each set of $\mathcal F$ contains an element of $S$}
{Question}

\textsc{Hitting Set} is unlikely to have a polynomial kernel when parameterized by $k+|U|$, unless the polynomial hierarchy collapses to its third level~\cite[Theorem 5.1]{DLS09journal}.

\begin{theorem}
\label{thm:BD-nopolykernel}
\BDS{} parameterized by solution cost $k$ does not admit a polynomial kernel even on layered DAGs of maximum finite distance~$2$, unless the polynomial hierarchy collapses to its third level.
\end{theorem}

\begin{proof}
We provide a reduction from {\sc HITTING SET}. %
It is shown in~\cite[Theorem 5.1]{DLS09journal} %(see also Exercise 15.4.10 of~\cite{CFK+15})\todo{enlever ?}
that if \textsc{Hitting Set} admits a polynomial kernel when parameterized by $|U|+k$ (a variant called \textsc{Small Universe Hitting Set}), then the polynomial hierarchy collapses to its third level.

We do the same reduction as the one from \textsc{Exact Cover by $3$-Sets} from Theorem~\ref{thm:dbds:npc}, except that the set $T$ of triples is replaced by $U$ and the set $X$ of elements is replaced by $\mathcal F$. We again obtain a DAG with three layers and maximum finite distance~$2$. The solution cost for the instance of \BDS{} is set to $|U|+k$, and the proof of validity of the reduction is the same.

Since this is clearly a \PTPT, the result follows.
\end{proof}

We now show that \BD{} can be solved in polynomial time on special kinds of DAGs.

\begin{theorem}
\label{thm:dbds:sld}
    \BDS{} is linear-time solvable on single-sourced layered DAGs.
\end{theorem}

\begin{proof}
    Let $D = (V,A)$ be a single-sourced layered DAG with layers $\{V_0,\ldots,V_t\}$. For the sake of readability, sets $V_i$ such that 
    $\size{V_i} = 1$ are denoted $\{s_i\}$, for $0 \leqslant i \leqslant t$. 
    
    Our algorithm relies on the following 
    structural properties of some optimal dominating broadcasts for single-sourced layered DAGs. 
    
    \begin{claim}
    \label{claim:dbds:sld}
        There always exists an optimal dominating broadcast $f$ of $D$ such that: 
        \begin{enumerate}[(i)]
            \item $V_f \subseteq \bigcup_{i=0}^t s_i$
            \item every $s_i \in V_f$, $0 \leqslant i \leqslant t$, covers exactly
            $B^+_{l}(s_i)$,
            where $l=j-i-1$ and $j$ is the smallest index such that $j\geq i+2$ and $\size{V_{j}} = 1$.
        \end{enumerate}
    \end{claim}
   
 \begin{proofclaim}
        Let $f$ be an optimal dominating broadcast of $D$ having the properties of 
        Proposition~\ref{prop:dbds:eff}.
        
        \medskip
        
        \noindent \textbf{Property $(i)$.} Let $0 \leqslant i < j \leqslant t$ be indices such that $s_i$ covers all layers 
        up to $V_{j-1}$, where $j$ is the smallest index such that $\size{V_j} \geqslant 2$ and 
        $f(V_j) > 0$. Notice that $i$ exists since $f(s_0) \geqslant 1$. 
        If $j$ does not exist, then we are done. We hence assume $j$ is well-defined. %,  
%        and let $v_{j-1}$ be any vertex of $V_{j-1}$ such that $v_j \in N^+(v_{j-1})$. 
        By the choice of $i$, we know that $f(s_i) = d(s_i, V_{j-1})=j-i-1$. % for any vertex $v_{j-1} \in V_{j-1}$. 
        Let $v^1_j$ and $v^2_j$ be two vertices of $V_j$. 
        We first consider the case where $\size{V_f \cap V_j} = 1$ 
        and assume w.l.o.g. that $f(v^1_j) \geqslant 1$. This means that $v_j^2$ must be covered by $s_i$, which in turn covers $v^1_j$, which is impossible by the choice of $i$ (and the definition of $f$). 
        We thus have $\size{V_j \cap V_f} \geqslant 2$, and assume that $f(v^1_j) \geqslant 1$ and 
        $f(v^2_j) \geqslant 1$. 
        Assume first that $j = t$. In that 
        case, $s_i$ covers all vertices in 
        $\cup_{a=i}^t V_{a-1}$, and hence setting $f(v^1_j) = f(v^2_j) = 0$ and increasing $f(s_i)$ by~$1$ leads to a dominating broadcast of smaller cost, 
        a contradiction. 
        %and hence setting $f(s_i) = f(s_i)+1$ and $f(v^1_j) = f(v^1_j) - 1$ %\todo{FF: si $j=t$ pourquoi on met pas $f(v^1_j)=f(v^2_j)=0$?} leads 
        %to an optimal dominating broadcast.%\todo{need to say that this has no impact on layers $0$ to $i-1$?}. 
        \medskip
        
        We thus assume $j < t$. We claim that the dominating broadcast $f_i$ defined by 
        setting: 
            \[ 
            \left \{
              \begin{array}{cccr}
              f_i(s_i) & = & f(s_i) + max\{f(v^1_j),f(v^2_j)\} + 1 & \\
              f_i(v^1_j) & = & 0 &\\
              f_i(v^2_j) & = & 0 &\\
              f_i(v) & = & f(v) & \forall\ v \neq \{s_i, v^1_j, v^2_j\}
              \end{array}
              \right.
            \]
        is optimal. Notice first 
        that $c_{f_i}(V) \leqslant c_f(V)$. Now, every vertex covered by both $v_j^1$ and $v_j^2$ is 
        covered by $s_i$: indeed, since $s_i$ corresponds to a layer with a single vertex, it has 
        a directed path of length $d(s_i,v_{j-1})+max\{f(v^1_j),f(v^2_j)\} + 1$ to every vertex  
        covered by both $v^1_j$ and $v^2_j$, which are thus still covered.
        
        \medskip
        
        \noindent \textbf{Property $(ii)$.} Suppose that $f$ satisfies Property~(i). Assume there exist two vertices $s_i$ and $s_j$ with $0 \leqslant i <i+1< j \leqslant t$ such that $f(s_i) \geqslant d(s_i, s_j)$. In other words, vertex $s_i$ covers 
        vertex $s_j$. Consider that $i$ is chosen to be minimum with this property. Notice that since $f$ fulfills the properties of Proposition~\ref{prop:dbds:eff}, 
        we have $f(s_j) = 0$. 
        We distinguish two cases:
        \begin{itemize}
            \item If $f(s_i) > d(s_i, s_j)$, consider the dominating broadcast $f_i$ obtained from $f$
        by setting $f_i(s_i) = d(s_i, s_j) - 1$ and $f_i(s_j) = f(s_i) - d(s_i, s_j)$. Notice that 
        every vertex covered by $s_i$ in $f$ is still covered in $f_i$: indeed, $s_i$ covers all vertices 
        up to $V_{j-1}$, and vertices in higher layers are now covered by $s_j$, which covers itself. 
        By construction, we have: 
        \begin{align*}
            c_{f_i}(V) & =  c_{f_i}(V \setminus \{s_i, s_j\}) + f_i(s_i) + f_i(s_j) \\
                       & =  c_{f_i}(V \setminus \{s_i, s_j\}) + d(s_i, s_j) - 1 + f(s_i) - d(s_i, s_j) \\
                       & <  c_f(V \setminus \{s_i, s_j\}) + f(s_i) \\
                       & <  c_f(V)
        \end{align*} 
        the last inequality %equality 
        holding since $f(s_j) = 0$. This leads to a contradiction since $f$ is an optimal dominating broadcast. Thus this case does not happen.
        \item We may hence assume that $f(s_i) = d(s_i, s_j)$. Since $f$ fulfills the properties of Proposition~\ref{prop:dbds:eff}
        and Property~$(i)$, $V_{j+1}$ has to dominate itself, and thus $s_{j+1}$ must exist, unless $j=t$.  Consider the dominating broadcast $f_i$ obtained from $f$
        by setting $f_i(s_i) = d(s_i, s_j) - 1$, $f_i(s_j) = 1+f(s_{j+1})$ and $f_i(s_{j+1})=0$. If $j=t$ we consider that  $f(s_{j+1})=0$. Notice that every vertex covered by $s_{j+1}$ in $f$ is covered by $s_j$ in $f_i$. We have:
        \begin{align*}
            c_{f_i}(V) & =  c_{f_i}(V \setminus \{s_i, s_j,s_{j+1}\}) + f_i(s_i) + f_i(s_j)+f_i(s_{j+1}) \\
                       & =  c_{f_i}(V \setminus \{s_i, s_j,s_{j+1}\}) + d(s_i, s_j) - 1 + f(s_{j+1}) +1 \\
                       & =  c_f(V \setminus \{s_i, s_j,s_{j+1}\}) + f(s_i) + f(s_{j+1})\\
                       & =  c_f(V)
        \end{align*} 
        the last %inequality 
        equality holding since $f(s_j) = 0$. 
        \end{itemize}
        
        We have thus obtained a dominating broadcast $f_i$ of the same cost as $f$, still satisfying Property~(i) and Proposition~\ref{prop:dbds:eff}, but where every vertex $s_l$ with $l\leq i$ satisfies~(ii). If $f_i$ still does not satisfy~(ii), we reiterate this process (each time, with increasing value of $i$) until  (ii) is satisfied for all vertices. This concludes the proof of Claim~\ref{claim:dbds:sld}.
    \end{proofclaim}
    
    We thus deduce a simple top-down procedure to compute an optimal dominating broadcast $f$. 
    We initiate our solution by setting $i=0$. While there remain uncovered vertices, we let $f(s_i) = j-i-1$ for the smallest value $j$ such that $s_j$ exists and $j\geq i+2$. In other words, $s_i$ will cover all vertices below it, until the closest vertex of the set $\bigcup_{j=0}^t s_j$ that is not a neighbour of $s_i$. We then carry on by setting
    $i=j$. By Claim~\ref{claim:dbds:sld}, this process leads to the construction of an 
    optimal dominating broadcast. 
 \end{proof}

\subsection{Algorithms for structural parameters and structured classes}

We now give some algorithms for structural parameters and classes.

\begin{theorem}
\BD{} can be solved in time $d^dn^{O(d)}$ for digraphs of order $n$ and diameter $d$.
\end{theorem}
\begin{proof}
To solve \BD{} by brute-force, we may try all the subsets of size $k$, and for each subset, try all possible $k^k$ broadcast functions. But we can assume that $k\leq d$, since a single vertex with cost $d$ covers all the digraph, which completes the proof. 
\end{proof}

%\todo{maybe too short for a subsection?}
We next consider jointly two parameters. Recall that by Theorems~\ref{thm:dbds:npc} and~\ref{thm:dbds:w}, such a result probably does not hold for each of them individually.

\begin{theorem}\label{thm:BD-degree}
\BD{} parameterized by solution cost $k$ and maximum out-degree $d$ can be solved in FPT time $k^{k}2^{d^{O(k)}}n^{O(1)}$ on digraphs of order $n$.
\end{theorem}

\begin{proof}
Let $(D = (V,A),k)$ be an instance of \BD{} such that $D$ has maximum out-degree~$d$. Consider a dominating broadcast $f$ of cost $k$. A vertex $v$ with $f(v)=i>0$ covers all vertices of its ball of radius~$i$, which has size at most $\sum_{j=0}^{i}(d-1)^j+1\leq id^i+1$. Thus, if the input has more than $n=k(k+1)d^k$ vertices, we can reject. Otherwise, a simple brute-force algorithm over all possible $2^{n}$ possible subsets and, given a subset, all $k^k$ possible broadcasts, is FPT. The result follows.
\end{proof}

Next, we consider the \emph{vertex cover number} of input digraphs, that is, the smallest size of a set of vertices that intersects all arcs (or, in other words, the vertex cover number of the underlying undirected graph).

\begin{theorem}\label{thm:BD-VC}
\BD{} parameterized by the vertex cover number $c$ of the input digraph of order $n$ can be solved in FPT time $2^{c^{O(c)}}n^{O(1)}$.
\end{theorem}
\begin{proof}
Let $(D=(V,A),k)$ be an instance of \BD{} and let $S \subseteq V$ be a vertex cover of $D$ of size $c$. Let us partition the set $V\setminus S$ (which contains no arcs) into equivalence classes $C_1,\ldots, C_t$ according to their in- and out-neighborhoods in $S$: two vertices are in the same class if and only if they have the same sets of in- and out-neighbors. There are $t\leq 2^{2c}$ such classes.

For a given class, any broadcasting vertex out of the class either covers all vertices in the class, or none. Similarly, a vertex broadcasting at radius $r$ inside the class covers the same set of vertices outside the class as any other vertex from the class would. Hence, we may assume that at most one selected vertex $b_i$ per class $C_i$ broadcasts with $f(b_i)>1$. We can assume that the other vertices in the class either all satisfy $f(v)=0$ or all $f(v)=1$ (the latter may happen if they all need to cover themselves, for example if they are all sources). 
Moreover, $\mfd(D)\leq 2c+1$ since every shortest path is either contained in $S$ or 
has to alternate between a vertex of $S$ and one of $V\setminus S$, but cannot have repeated vertices.

Hence, for each equivalence class $C_i$, we have $2\times(2c+1)$ choices: $2c+1$ for the value of $f(b_i)$, and two possibilities for the other vertices of $C_i$. Similarly, for each vertex of $S$, we have $2c+1$ possible broadcast values. In total, this gives $(t+c)^{O(c)}=2^{c^{O(c)}}$ different possible dominating broadcasts, and each of them can be checked in polynomial time. 
\end{proof}

We next see how to apply the following powerful theorem from~\cite{FOnowheredense}, to show that \BD{} is FPT for any class of digraphs whose underlying graph is \emph{nowhere dense}. We will not give a definition of nowhere dense graph classes, and refer to the book~\cite{SparsityBook} instead. Such classes include planar graphs, graphs excluding a fixed (topological) minor, graphs of bounded degree, graph classes of bounded expansion, etc.

\begin{theorem}[\cite{FOnowheredense}]\label{thm:FOnowhere-dense}
Let $\mathcal C$ be a nowhere dense graph class. There exists $\epsilon$ such that, given as inputs a graph $G\in\mathcal C$ and a first-order logic graph property $\varphi$, the problem of deciding whether $G$ satisfies $\varphi$ can be solved in time $f(|\varphi|)|G|^{1+\epsilon}$, that is, it is FPT when parameterized by the length of $\varphi$.
\end{theorem}

\begin{corollary}\label{thm:BD-nowhere-dense}
For every fixed nowhere dense graph class $\mathcal C$, \BD{} parameterized by the solution cost of the input digraph is FPT for inputs whose underlying graphs are in $\mathcal C$.
\end{corollary}
\begin{proof}
We want to show that for fixed parameter value~$k$ of the solution cost, \BD{} can be expressed in first-order logic by a formula whose length is bounded by a function of $k$, and apply Theorem~\ref{thm:FOnowhere-dense}.

To do so, we extend the classic approach for defining \textsc{$k$-Dominating Set} in first-order logic (see e.g. \cite[Chapter~18.4]{SparsityBook}).

We will use the property $dp(x,y,i)$, stating that there is a directed path from $x$ to $y$ of length at most~$i$. This can be expressed in first-order logic for fixed~$i$. To this end, we state that either $x=y$, or there is an arc from $x$ to $y$, or there is a directed path of length~$2$ from $x$ to $y$ (i.e. there exists a vertex $z$, $x$ has an arc to $z$, and $z$, an arc to $y$), $\ldots$ or there exists a directed path of length~$i$ from $x$ to $y$.

Let $V_1,\ldots,V_k$ denote the sets of broadcast dominators of a potential dominating broadcast $f$, where $V_i$ contains the vertices broadcasting at radius~$i$. The union $V_f=\bigcup_{i=1}^k V_i$ has size at most $k$, and since $k$ is considered to be fixed, we can ''guess'' the size of each set $V_i$. To this end, we let $v_i^1,\ldots,v_i^k$ be the potential vertices of $V_i$. For a given partition $\Pi$ of $V_f$ into sets $V_1,\ldots,V_k$, we can express the fact that a given vertex $x$ is dominated by $f$ as the formula $dom_\Pi(x,v_1^1,\ldots,v_k^k)$, which is composed of the conjunction of all formulae of type $dp(v_i^j,x,i)$, where in $\Pi$, $1\leq j\leq |V_i|$.

Now, given the set $\Pi_1,\ldots,\Pi_t$ of all partitions of $V_f$ into sets $V_1,\ldots,V_k$ (note that $t\leq k^k$), the first-order formula for \BD{} is given as $$\exists v_1^1\ldots\exists v_k^k, \left(\forall x\in G, dom_{\Pi_1}(x,v_1^1,\ldots,v_k^k)\right)\vee\ldots\vee\left(\forall x\in G, dom_{\Pi_t}(x,v_1^1,\ldots,v_k^k)\right).$$
\end{proof}

We remark that Corollary~\ref{thm:BD-nowhere-dense} does not imply Theorem~\ref{thm:BD-degree}, indeed there are digraph classes of bounded maximum out-degree whose underlying graphs do not form a nowhere dense class of graphs. For example, every $d$-degenerate graph can be oriented so as to have maximum out-degree at most~$d$. Indeed, a graph is \emph{$d$-degenerate} if its vertices can be ordered $v_1,\ldots,v_n$ such that for $2\leq i\leq n$, $v_i$ has at most $d$ neighbors among $v_1,\ldots,v_{i-1}$. Thus, orienting every edge $v_iv_j$ with $i<j$ from $v_j$ to $v_i$ produces a digraph of maximum out-degree at most~$d$. However, for every $d$, the class of $d$-degenerate graphs is not nowhere dense~\cite{FOnowheredense,SparsityBook}.

\section{Complexity of \MP}\label{sec:MP}

We will need the following results to prove our results for \MP{}. The first one was proved for undirected graphs in~\cite{HM14}.

\begin{lemma}\label{lemm:MPlongpath}
Let $D = (V,A)$ be a digraph with a shortest 
directed path of length $3k-3$ vertices. Then, $D$ has a multipacking of size $k$.
\end{lemma}

\begin{proof}
It suffices to select every third vertex on the path.
\end{proof}

\begin{lemma}
\label{lem:mp:sources}
    Let $D = (V,A)$ be a digraph. 
    There always exists a multipacking of maximum size containing every source of $D$. %
\end{lemma}

\begin{proof}
    Let $D=(V,A)$ and let $S \subseteq V$ be a multipacking of $D$ of size at least $k$.
    % FIXME: easier to say: move a vertex from the first non-empty (w.r.t. S) $N_d^+(s)$?} 
    Assume there exists a source $s \in V$ that does not belong to $S$. 
    We say that a vertex $v \in V$ is \emph{full} w.r.t. $S$ whenever there exists an integer $p > 0$ such that 
    $|B^+_p(v) \cap S| = p$. 
    Assume first that $s$ is not full w.r.t. $S$: in that case, one can safely add $s$ to the 
    multipacking $S$ and obtain a new solution of size at least $k$. Hence, we now consider the 
    case where $s$ is full. Notice that if $s$ is full at distance~$1$ (i.e. 
    $\size{B^+_1(s) \cap S} = 1$), then the set 
    $(S \setminus \{u\}) \cup \{s\}$ is a multipacking of size at least $k$ (recall that $s$ is a source), and thus we are done. \\
    
    We hence assume that this is not the case. 
    Let $1 \leqslant i \leqslant \ecc(s)$ be the smallest integer such that $|B^+_i(s) \cap S| < i$ and $|B^+_{i+1}(s) \cap S| = i+1$. 
    %The case $i = 1$ means that $N^+[s]$ contains exactly one vertex $u$ from $S$. T
    Notice that $|N^+[s] \cap S| = 0$, since otherwise $s$ would be full at distance~$1$. In 
    particular, since $s$ is full at distance~$i+1$, this means that $|B_{i+1}^+(s) \cap S| \geqslant 2$. 
    Let $u$ be any vertex of $B_{i+1}^+(s) \cap S$. We claim that the set $S' = (S \setminus \{u\}) \cup \{s\}$ is a multipacking of $D$. 
    First, it is clear that $|S'| = |S|$. Now, since $s$ is a source and $\size{N^+(s) \cap S} = 0$, 
    adding $s$ to the multipacking 
    cannot violate the constraint for any vertex $v \in V$. Similarly, removing a vertex from a 
    multipacking cannot create any new constraint, hence the result follows. 
\end{proof}

The following lemma is the central result of both our polynomial-time algorithm (Theorem~\ref{lem:mp:polysld}) and NP-completeness reduction 
(Theorem~\ref{thm:dmp:sdnp}).  

\begin{lemma}
\label{lem:mp:sld}
    Let $D = (V,A)$ be a single-sourced layered DAG. 
    There always exists a multipacking $S \subseteq V$ of maximum size such that for every $1\leq i\leq t$, $|S \cap V_i| \leqslant 1$. 
\end{lemma}

\begin{proof}
    Let $S \subseteq V$ be a multipacking of $D$ of maximum size. 
    By definition of a multipacking, considering each ball centered at the source $s$, the following holds for every $1\leqslant i \leqslant t$:
    \begin{equation}
    \label{eq:mp:sld}
        \size{S \cap \cup_{j =0}^i  V_j} \leq i
    \end{equation}
    We will prove the result inductively, by locally modifying $S$ in a top-down manner until it has the desired property. Let $j\geq 2$ be the smallest index such that $\size{S \cap V_j} \geq 2$, and $i < j$ be 
    the largest index such that $\size{S \cap V_i} = 0$. Notice that $i$ is well-defined due to~\eqref{eq:mp:sld}. Moreover, let $s_j^1$ and $s_j^2$ be two vertices of $S \cap V_j$. 
    
    \medskip
    
    \noindent \textbf{Case 1.} We assume first that $i = j-1$. Let $u_i^1$ and $u_i^2$ be vertices of $V_i$ such that 
    $u_i^1s_j^1$ and $u_i^2s_j^2$ belong to $A$ (note that in a layered DAG every non-source vertex has a predecessor in the previous layer). Since $S$ is a multipacking, 
    we have $u_i^1 \neq u_i^2$ and neither $u_i^1$ nor $u_i^2$ is adjacent 
    to both $s_j^1$ and $s_j^2$. Moreover, a vertex $s_{i-1}$ in $S \cap V_{i-1}$ cannot be 
    adjacent to both $u_i^1$ and $u_i^2$, since otherwise we would have $\size{B^+_2(s_{i-1})\cap S} > 2$. Moreover by minimality of the index $j$, there is at most one vertex of $S$ in $V_{i-1}$.
    Assuming w.l.o.g. that $u^1_i$ has no predecessor in $S$, the set 
    $(S \setminus \{s_j^1\}) \cup \{u_i^1\}$ is a multipacking having the same size than $S$. 
    
    \medskip
    
    \noindent \textbf{Case 2.} We now consider the case where $i < j-1$. First, we will prove that there is a vertex $v_i$ in $V_{i}$ with no in-neighbor in $S$. If $S \cap V_{i-1} = \emptyset$, any vertex of $V_i$ can be chosen as vertex $v_i$. 
    Otherwise, by choice of $j$ we have $\size{S \cap V_{i-1}}=1$. Assume $S \cap V_{i-1} = \{s_{i-1}\}$. %
    We claim that $s_{i-1}$ is not 
    adjacent to every vertex of $V_i$. Assume for a contradiction that this is the case. 
    This means that $s_{i-1}$ is within distance $j-(i-1)$ of every vertex contained in $\cup_{l=i}^j V_l$. 
    By the choice of indices $i$ and $j$ we know that $\cup_{l=i}^j V_l$ contains at least 
    $j-(i-1)$ vertices from $S$, which in turn implies that $\size{B^+_{j-(i-1)}(s_{i-1})\cap S} = j-(i-1)+1$, 
    contradicting~\eqref{eq:mp:sld}. 
    Thus, there is a vertex $v_i$ in $V_i$ that has no in-neighbor in $S$. 
    Now, we know by choices of $i$ and $j$ that $|S \cap V_p| = 1$ for $i < p < j$. Hence the set 
    $(S \setminus \{s_{i+1}\}) \cup \{v_i\}$, where $\{s_{i+1}\} = S \cap V_{i+1}$, is a multipacking of $D$ 
    having the same size than $S$. %
    By iterating the above argument, we end up with $i = j-1$, in which case we can apply the argument from Case~1. 
    Overall, after each iteration of Case~1, $j$ strictly increases. %
    The procedure terminates when the value of $j$ reaches $t$.
\end{proof}

\subsection{Hardness results}\label{sec:hard-MP}

\begin{theorem}
\label{thm:dmp:npc}
\MP{} is NP-complete, even for planar layered DAGs of maximum degree~$3$ and maximum finite distance~$3$.
\end{theorem}

\begin{proof}
We provide a reduction from the NP-complete \textsc{Independent Set} problem~\cite{GJ79}, 
    which remains NP-complete on planar cubic graphs~\cite{GJS76}.
    
\Pb{\textsc{Independent Set}}
{A graph $G=(V,E)$, an integer $k\in\mathbb{N}$.}
{an independent set of $G$ of size at most $k$}
{Question}
    
The construction of the instance $(D = (V',A'), k')$ of \MP{} 
is done by setting $V' = E_1 \cup E_2 \cup V$ where $E_1 = \{e_1^1, \ldots, e_m^1\}$ and 
$E_2 = \{e^2_1, \ldots, e^2_m\}$ are two copies of $E$. 
We add an arc $e^1_ie^2_i$ for every $1 \leqslant i \leqslant m$, and two arcs from $e_i^2$ 
to the corresponding vertices $u$ and $v$ in $V$ (where $e_i=uv$).

%See Figure~\ref{fig:mp:npc} for an illustration.
%
%\begin{figure}[ht!]
%    \centerline{\includegraphics[scale=1.5,angle=90,origin=c]{npc-mp}}
%    \caption{Sketch of the construction of the digraph $D$ in the proof of Theorem~\ref{thm:dmp:npc}, for an edge $e_i=uv$.}\label{fig:mp:npc}
%\end{figure}

Formally: 
    $$A' =  \{e^1_ie^2_i:\ 1 \leqslant i \leqslant m \} \cup \{e^2_iu,\ e^2_iv:\ 1 \leqslant i \leqslant m\ \textstyle{and}\ e_i = uv\}$$

\fullversion{It is clear here that $D$ is a layered DAG with three layers and thus, $\mfd(D)=2$. }
\noindent This reduction can also be seen as follows: 
given any instance of \textsc{Independent Set}, we subdivide each edge $uv$ by adding 
a new vertex $w$ with $wu, wv \in A$ and a pending 
source seeing $w$. Doing so, most properties of the given instance (such as planarity and maximum degree) are preserved. One can see that the graph $G$ has an independent set of size $k$ if and only if the digraph $D$ has a multipacking of size $k'= m+k$. %

    $\Rightarrow$ Let $S$ be an independent set of $G$ of size $k$, and let $S' = E_1 \cup S$. First, $S'$ is of size $m+k$. Then, for any $e_1\in E_1$, $|N^+[e_1]\cap S'| =1$ and $|B_2^+(e_1) \cap S'| \leq 2$ hold since $S$ is an independent set.
    By similar arguments, $|N^+[e_2]\cap S'| \leq 1$ holds for any $e_2 \in E_2$, and thus no vertex 
    of $E_2$ can have two out-neighbors in $S$. 
    All other vertices of $D$ are sinks (i.e. with empty out-neighborhood), so the multipacking property is trivially satisfied for them. Thus $S'$ is a multipacking of $D$ of size $m+k$. \\
    $\Leftarrow$ Let $S$ be a multipacking of maximum size in $D$, such that $|S| \geq m+k$.
    Each vertex of $E_1$ is a source of $D$, so by Lemma~\ref{lem:mp:sources} we can assume that $E_1 \subseteq S$ and 
    then $E_2 \cap S = \emptyset$. So $S \setminus E_1 \subseteq V$, and its size is at least $k$. Assume $S$ contains two vertices $u,v$ of $V$ that are adjacent in $G$, then $\size{N^+[e_i^2] \cap S}\geq 2$ with $e^2_i = uv$, which contradicts the fact that $S$ is a multipacking of $D$. Thus $S\setminus E_1$ is an independent set of $G$ of size at least $k$.
 \end{proof}

\begin{remark}\label{rem:mp-eth}
\MP{} can be solved in time $O^*(2^n)$ by trying all subsets of vertices as a solution. 
By observing that the reduction of Theorem~\ref{thm:dmp:npc} from \textsc{Independent Set} is linear and that it is unlikely to obtain a subexponential algorithm for \textsc{Independent Set} under the ETH~\footnote{The Exponential Time Hypothesis (ETH) assumes that there is no algorithm solving 3-SAT in time $2^{o(n)}$, where $n$ is the number of variables in the formula.}~\cite[Corollary 11.10]{FominK10}, a subexpontential algorithm is also unlikey for \MP{} under the ETH.
\end{remark}

\begin{theorem}
\label{thm:dmp:sdnp}
 \MP{} is NP-complete on single-sourced DAGs of maximum degree~$5$. %
\end{theorem}

\begin{proof}
    We provide a reduction from \textsc{Independent Set} problem~\cite{GJ79}, which remains NP-complete for cubic graphs~\cite{GJS76}.  
    %where one is given an 
    %undirected graph $G = (V,E)$ and seeks an independent set $S \subseteq V$ of size at least $k$.  
We %consider some order $\{e_1, \ldots, e_m\}$ on the edges of $G$ and 
define the function $f:V\rightarrow E$ such that for $v\in V$, $f(v)=e_i$ if and only if $e_i$ is the first edge in which $v$ appears (recall that $E = \{e_1, \ldots, e_m\}$). We create the digraph $D = (V',A)$ as follows (see Figure~\ref{fig:dmp:npcslay2}): 
\begin{align*}
    V' = & \{u_i,v_i,w_i,x_i,y_i,z_i:\ 1 \leqslant i \leqslant m \} \cup V \cup \{s,p\} \\
    A = & \{u_iw_i, u_ix_i:\ 1 \leqslant i \leqslant m\} \cup \{v_ix_i:\ 1 \leqslant i \leqslant m\}\cup \{w_iy_i,w_iz_i:\ 1 \leqslant i \leqslant m\}\bigcup\\  
      &    \{z_iu_{i+1},z_iv_{i+1}:\ 1 \leqslant i \leqslant m-1\}\cup\{x_iu,x_iv:\ 1 \leqslant i \leqslant m \text{ and } e_i = uv\} \bigcup \\
      &    \{u_iu:\ 1 \leqslant i \leqslant m \text{ and } f(u)=e_i \}\cup \{sp,pu_1,pv_1\}
\end{align*}

\begin{figure}[ht!]
\centering
\includegraphics[scale=1.5]{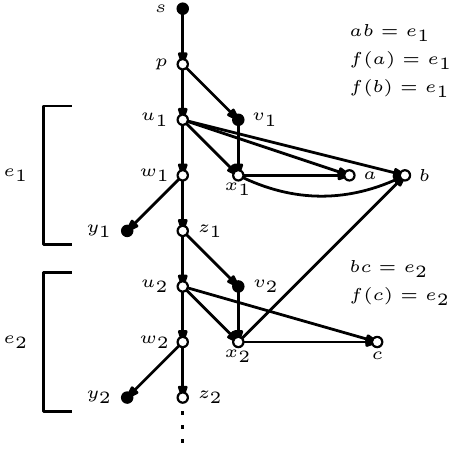}
\caption{Sketch of the construction in the proof of Theorem~\ref{thm:dmp:sdnp} for edges $e_0 = ab$ and $e_1 = bc$ with $f(a) = f(b) = e_0$ and $f(c) = e_1$.}
\label{fig:dmp:npcslay2}
\end{figure}

\begin{claim}
    The graph $G$ has an independent set of size $k$ if and only if the digraph $D$ 
    has a multipacking of size $k' = k + 2m+1$.
\end{claim}

\begin{proofclaim}
    $\Rightarrow$ Let $S$ be an independent set of size $k$ of $G$. We set $S' = \{s\}\cup \{v_i,y_i:\ 1 \leqslant i \leqslant m\}\cup S$. We need to show that $S'$ is a multipacking of $D$. Notice 
    first that $S'$ contains exactly $2m + k + 1$ vertices. 
    %Although $D$ is not a layered DAG, the subdigraph of $D$ induced by $V'\setminus V$ is a layered DAG. We consider these layers, and we also add every vertex $v\in V$ to every layer in which we find $w_j,x_j$ (for every $j$ such that $v$ is an endpoint of $e_j$). Thus these new ''layers'' are not vertex-disjoint.
    The vertices $s$ and $p$ satisfy the multipacking property since there is at most one vertex of $S'$ 
    at distance exactly $i$ from both these vertices, for any $i$ (and there is no vertex of $S'$ at distance~$1$ from $s$ 
    and none at distance~$0$ from $p$). Each vertex of $V$ and each vertex $y_i$ trivially satisfies the 
    multipacking property since they are sinks. 
    For $1\leqslant i \leqslant m$, notice that $x_i$ cannot have two out-neighbors in $S'$ 
    since $S$ is an independent set. Hence, $x_i$ and $v_i$ satisfy the multipacking property, since for the latter 
    $B^+_d(v_i) = \{v_i,x_i,u,v\}$ where $d$ is the maximum finite distance in $D$, $uv=e_i$, and $N^+[v_i] = \{v_i,x_i\}$. 
    %Finally,     $y_i$ satisfies trivially the multipacking property since it is a sink. 
    Moreover, one can see that $w_i$ satisfies the multipacking property if and only if $z_i$ satisfies it and that $z_i$ satisfies the multipacking property if and only if $u_{i+1}$ satisfies it ($z_m$ is a sink, hence satisfies the multipacking property).
    We can notice that $B^+_d(u_i) = B^+_d(w_i) \cup \{x_i,u_i\}\cup V(e_i)$. We have $\size{S\cap (\{x_i,u_i\}\cup V(e_i))} \leqslant 1$, and the fact that for every other vertex $t$ of $B^+_d(u_i)$, $d(u_i,t)=d(w_i,t)+1$. So if $w_i$ satisfies the property, then $u_i$ also does. This means that $z_{i-1}$ satisfies it, and thus that $w_{i-1}$ does as well. Using this, and the fact that $z_m$ satisfies the property, we get by induction that for every~$i$, $\{u_i,w_i,z_i\}$ satisfy the property.\\ 
    %One can easily see that there is still at most one vertex of $S$ per above-defined layer, which gives us the fact that $u_i$ satisfies the multipacking property.  Indeed, when considering 
    %the layers, the maximal ball $B^+_d(u_i)$ centered at $u_i$ contains only vertices in the lower layers (and it never skips any). \\
    $\Leftarrow$ Let $S$ be a multipacking of size $k'$ of $D$. First, notice that if $M$ is a multipacking of any digraph $H$, then for any subdigraph $H'$  of $H$, $M\cap V(H')$ is a multipacking of $H'$.
    Notice also that $H=D[V'\setminus V]$ is a single-sourced layered DAG. Let $S'$ be a multipacking of $H$
    of maximum size. Using Lemma~\ref{lem:mp:sld}, we can assume that $S'$ contains at most one vertex per
    layer. For any given $1\leqslant i\leqslant m$, we are going to prove that for
    $W_i=\{u_i,v_i,w_i,x_i,y_i,z_i\}$, $\size{S'\cap W_i}\leqslant 2$. We can see that $S'\cap \{u_i,v_i\}$ is
    either empty (which is sufficient to conclude since there remains only two distinct nonempty layers of $D'$ in $W_i$), or $S'\cap \{u_i,v_i\}=u_i$ (then $S'\cap
    \{w_i,x_i\}=\emptyset$, which again is enough to conclude), or $S'\cap \{u_i,v_i\}=v_i$. In the latter case,
    either $S'\cap \{w_i,x_i\}=w_i$, which implies that $S'\cap \{y_i,z_i\}=\emptyset$ or $S'\cap
    \{w_i,x_i\}=\emptyset$. In both cases, we get that $\size{S'\cap W_i}\leqslant 2$. One can also easily see
    that both $s$ and $p$ cannot be together in $S'$.
    Thus, the maximum size of a multipacking of $D'$ is $2m+1$. 
    
    Thus $\size{S\cap 
    (V'\setminus V)}\leqslant 2m+1$, and $\size{S\cap V}\geqslant k$. We also know that for $a,b \in S 
    \cap V$, $ab \notin E$, otherwise there would exist an edge $e_i = ab$ and thus $N^+[x_i] \cap S$ 
    would be of size at least $2$. So we can conclude that $S\cap V$ is an independent set of $G$ of size 
    at least $k$.    
    %\todo[inline]{We use the algorithm provided for single-source layered dags, so we need to put this after it, but it doesn't help with the structure....}
\end{proofclaim}
This completes the proof.
\end{proof}

\begin{theorem}
\label{thm:dmp:w}
\MP{} parameterized by solution size $k$ is W[1]-hard, even on digraphs of maximum finite distance~$3$.
\end{theorem}\label{thm:mp:whard}

\begin{proof}

We provide an \FPTR{} from \MIS{}, which is W[1]-hard when parameterized by $k$~\cite{CFK+15}.

\Pb{\MIS{}}
{A graph $G=(V,E)$ with $V$ partitioned into sets %$k$ sets 
$\{V_1, \ldots, V_k\}$, $k \in \mathbb{N}$.}
{an independent set $S$ of $G$ s.t. $|S \cap V_i |=1$ for $1 \leqslant i \leqslant k$}
{Question} 

%where one is given a graph $(G = (V,E), k)$ of \MIS{}, with $V$ partitioned into $\{V_1, \ldots, V_k\}$, and seeks an independent set of size $k$ such that $\size{S \cap V_i} = 1$ for $1 \leqslant i \leqslant k$. 

\noindent \textbf{Construction.} We construct an instance $(D = (V',A'), k')$ of \MP{} as follows. 
We consider the bipartite incidence graph of $G$, that is we add $V \cup E$ to $V'$. 
To construct $A'$, we add an arc from a vertex $e \in E$ to a vertex $v \in V$ if and only if $e$ contains $v$. 
We next group vertices of $E$ into ${k \choose 2}$ sets $E_{i,j}$, $1 \leqslant i < j \leqslant k$
according to the colors of their corresponding endpoints, and  
add every possible arc within each set $E_{i,j}$. 
We next duplicate the vertices of each set $V_i$ into a set $V'_i$ such that there is an arc from 
each vertex $v_i \in V_i$ to its corresponding copy $v'_i$ in $V'_i$. 
Finally, we add $k$ vertices $\{s_1, \ldots, s_k\}$ such that there is an arc from $s_i$ to 
every vertex of $V_i$. Notice in particular that the maximum finite distance is~$3$.

\fullversion{See Figure~\ref{fig:mp:w} for an illustration.}

\begin{figure}[ht!]
    \centerline{\includegraphics[scale=1.5]{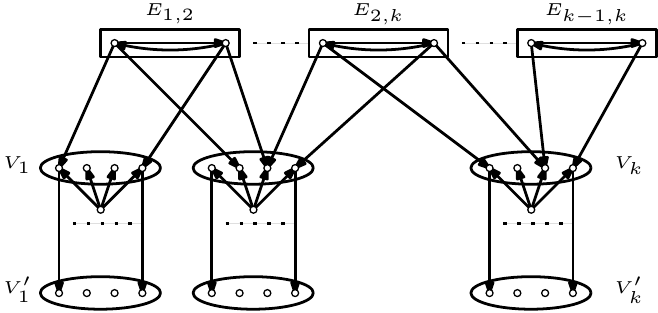}}
    \caption{Sketch of the construction of the digraph $D$ in the proof of Theorem~\ref{thm:mp:whard}.}\label{fig:mp:w}
\end{figure}

\begin{claim}
\label{claim:dmp:w}
    The graph $G$ has a multicolored independent set of size $k$ if and only if the digraph $D$ 
    has a multipacking of size $k' = 2k + {k \choose 2}$.
\end{claim}

\begin{proofclaim} $\Rightarrow$ Let $S = \{u_1, \ldots, u_k\}$ be an independent set of $G$ of size $k$ such that 
    $u_i \in V_i$ for every $1 \leqslant i \leqslant k$. Let $S' \subseteq V'$ 
    be a set that contains exactly one arbitrary vertex $e_{i,j}^*$ for every set $E_{i,j}$ ($1 \leqslant i < j \leqslant n$), 
    together with each vertex of $V'_i$ corresponding to each vertex $u_i$ of $S$. Finally, add 
    $\{s_1, \ldots, s_k\}$ to $S'$. 
    We claim that $S'$ is the sought multipacking of $D$. To see this, notice first that $\size{S'} = 2k + {k \choose 2}$ by construction. Moreover, 
    every vertex contains at most one vertex from $S'$ in its 
    closed out-neighborhood. We now prove that every vertex $e_{i,j} \in E_{i,j}$ contains at most two vertices 
    from $S'$ in $B^+_2(e_{i,j})$. Assume for a contradiction this is not the case; then, apart from $e_{i,j}^*$, there are two other vertices $a$ and $b$ in $B^+_2(e_{i,j})$. We have that $a \in V'_{i}$ and $b \in V'_{j}$. By 
    construction, this means that $ab$ is an edge of $G$, contradicting the fact that $S$ is an 
    independent set. Finally, since every vertex $s_i$ ($1 \leqslant i \leqslant k$) has  
    vertices from only one set $V'_i$ in its distance~$2$ neighborhood, and since $S$ is a multicolored set, 
    the result follows. The only vertices for which checking their distance~~$3$ neighborhood is needed are vertices from $E_{i,j}$ for every $1 \leqslant i < j \leqslant n$. One can notice that for any $e_{i,j} \in E_{i,j}$, $B^+_3(e_{i,j}) \subseteq E_{i,j} \cup V'_i \cup V'_j \cup V_i \cup V_j$, which contains at most $3$ vertices of $S'$ since $|S' \cap (V'_i \cup V'_j \cup V_i \cup V_j))| = 2$ and $|S' \cap E_{i,j}| = 1$ by construction. \\ %$S$ is a multicolored set.\\
    $\Leftarrow$ Assume that $D$ has a multipacking $S' \subseteq V'$ of size $k' = 2k + {k \choose 2}$. 
    By Lemma~\ref{lem:mp:sources}, we can assume that $S'$ contains $\{s_1, \ldots, s_k\}$. 
    In particular, this means that $S' \cap V_i = \emptyset$ for every $1 \leqslant i \leqslant k$.  
%    since otherwise for some vertex $s_i$ there would be two vertices from $S'$ in $B^+_2(s_i)$. 
    Moreover, at distance~$2$, we have $|S' \cap V'_i| \leqslant 1$ for $1 \leqslant i \leqslant k$ 
    since otherwise there would be three vertices from $S'$ in $B^+_2(s_i)$, for some vertex $s_i$. Moreover, for $1 \leqslant i < j \leqslant n$, $|E_{i,j} \cap S'| \leqslant 1$ since $E_{i,j}$ is a bi-directed clique.
    Thus, by the size of $S'$, the only possibilities are 
    to pick exactly one vertex in each set $V'_i$ and one vertex $e_{i,j}$ in each set $E_{i,j}$. 
    This can be done only if there exists a multicolored independent set of size $k$ in $G$: 
    otherwise one would have to select two vertices $a \in V_i$ and $b \in V_j$, $i \neq j$ such that 
    $ab \in E$, which in turn would imply that the vertex from $E_{i,j}$ corresponding to the edge $ab$ has three vertices in its
    distance-$2$ neighborhood (namely $e_{i,j}$, $a$ and $b$). 
\end{proofclaim}
Thus, the proof is complete.
 \end{proof}

\subsection{Algorithms}

Next, we present a linear-time algorithm.

\begin{theorem}
\label{lem:mp:polysld}
    \MP{} can be solved in linear time on single-sourced layered DAGs. 
\end{theorem}

\begin{proof}
    Let $D = (V,A)$ be a single-sourced layered DAG. By Lemma~\ref{lem:mp:sld}, in every single-sourced layered DAG there is a multipacking of maximum size 
    that is a maximum-size set of vertices with at most one vertex per layer such that two chosen vertices of 
    consecutive layers are not adjacent. %
    We thus give a polynomial-time bottom-up procedure to find such a set of vertices. At each step of 
    the procedure, a layer $V_i$ is partitioned into a set of \emph{active} vertices and a set of 
    \emph{universal} ones, denoted respectively $A_i$ and $U_i$. Our goal will be to select exactly one vertex in each set of active vertices. We initiate the algorithm by 
    setting $A_t = V_t$ and $U_t = \emptyset$. Now, for every $i$ with $0 \leqslant i < t$, we set 
    $U_i = \{u \in V_i:\ A_{i+1} \subseteq N^+(u)\}$ and $A_i = V_i \setminus U_i$. In other words, 
    $U_i$ contains the vertices of layer $V_i$ that are adjacent to all active vertices of $V_{i+1}$. 
    During the procedure, if some layer $V_i$ satisfies $A_i = \emptyset$, we let $A_{i-1} = V_{i-1}$ 
    and repeat this process until $V_0$ is reached.
    
    \medskip
    
    To construct a multipacking of maximum size, we start from $V_0$, and for each $0 \leqslant i \leqslant t$ we pick a vertex $s_i$ 
    in each non-empty set $A_i$ of active vertices. Every time a vertex $s_i$ is picked, we remove 
    its closed neighborhood from $D$. Notice that by construction, every time a vertex $s_i$ is 
    picked, there exists a vertex $s_{i+1} \in A_{i+1}$ such that $s_is_{i+1}$ does not belong to $A$  %\todo{rappeler $D$ et $A$ en début de preuve?}
    (otherwise $s_i$ would belong to $U_i$). 
    %Roughly speaking, we will show that our algorithm produces an optimal solution by showing that for every layer that does not contain a vertex from the solution (the layers with $A_i=\emptyset$), 
    %no optimal solution can have %\todo{il manque un mot?} 
    %more vertices in such a layer.
    To prove the optimality of our algorithm, let $0 \leqslant i < t$ be such that $A_i = \emptyset$, and $j>i$ be the smallest integer greater than $i$ such that $V_j = A_j$. Such a $j$ exists since $A_t=V_t$. 
    
    \begin{claim}
        \label{claim:perlayer}
        Let $S$ be a multipacking with at most one vertex per layer. Then $S$ satisfies:
    \begin{equation}
    \label{eq:mp:polysld}
        \size{S \cap \cup_{k =i}^j  V_k} \leqslant j-i
    \end{equation}
    \end{claim}

    \begin{proofclaim}
        Let $S$ be an optimal multipacking with at most one vertex per layer. Assume by contradiction that $\size{S \cap \bigcup_{k =i}^j  V_k} > j-i+1$, %
        and call $s_k$ the vertex in $V_k\cap S$ for every $i\leqslant k \leqslant j$. We know that $s_i \in U_i$, %
        and since every vertex in $A_{i+1}$ is an out-neighbor of $s_i$, then $s_{i+1} \in U_{i+1}$. %
        By induction, for every $i \leqslant k\leqslant j$, we have $s_k \in U_k$, but $U_j = \emptyset$ by choice of $j$, leading to a contradiction. %
    \end{proofclaim}

        Notice that Claim~\ref{claim:perlayer} gives one less vertex than what Lemma~\ref{lem:mp:sld} implies, 
        and that it is the value reached by our algorithm, since for $i\leqslant k \leqslant j$ the only layer with $U_k = V_k$ is $V_i$.
    Since the sets of active and universal vertices can be 
    constructed by standard graph searching, the whole algorithm takes $O(|V|+|A|)$ time. 
\end{proof}

We now give algorithms for structural parameters. We next give a simple algorithm for digraphs of bounded diameter.

\begin{theorem}
\MP{} can be solved in time $n^{O(d)}$ for digraphs of order $n$ and diameter $d$.
\end{theorem}
\begin{proof}
To solve \MP{} by brute-force, we may try all the subsets of size $k$, and for each subset, check its validity. But in a YES-instance, we have $k\leq d$, since any ball of radius $d$ contains all vertices.
\end{proof}

The next algorithm considers jointly two parameters. Recall that by Theorems~\ref{thm:dmp:npc} and~\ref{thm:dmp:w}, such a result probably does not hold for each of them individually.

\begin{theorem}
\MP{} parameterized by solution size $k$ and maximum out-degree $d$ can be solved in FPT time $d^{k^{O(k)}}n^{O(1)}$ for digraphs of order $n$.
\end{theorem}

\begin{proof}
Let $(D = (V,A), k)$ be an instance of \MP{} such that $D$ has maximum out-degree $d$. 
By Lemma~\ref{lemm:MPlongpath}, if $D$ has a shortest directed path of length $3k-3$, we can accept the input (this can be checked in polynomial time). Thus, we can assume that the length of any shortest path is at most $3k-2$. If a vertex $u$ has a directed path to a vertex $v$, we say that $u$ \emph{absorbs} $v$, and a set $S$ of vertices is \emph{absorbing} if every vertex in $D$ is absorbed by some vertex of $S$. If $D$ has a set of $k$ vertices, no two of which are absorbed by some common vertex (\emph{e.g.} a set of $k$ sources), we can accept, since this set forms a valid solution. %
Note that this property is satisfied by any minimum-size absorbing set $S$: indeed, if some vertex $w$ absorbs two vertices $u,v$ of $S$, we may replace them by $w$ and obtain a smaller absorbing set, a contradiction. %

We claim that we can find a minimum-size absorbing set in FPT time. Indeed, we can reduce this problem to \textsc{Hitting Set} (defined for the proof of Theorem~\ref{thm:BD-nopolykernel}) %
as follows. We let $U=V(D)$, and $\mathcal F$ contains a set $F_v$ for every vertex $v$, where $F_v$ comprises every vertex which absorbs $v$ (including $v$ itself). Because $D$ has out-degree at most $d$ and the length of any shortest path is at most $3k-2$, every vertex of $U$ is contained in at most $d_U=\sum_{i=0}^{3k-2}(d-1)^i+1$ sets of $\mathcal F$. Moreover, a set of vertices of $U=V(D)$ is a hitting set of $(U,\mathcal F)$ if and only if it is an absorbing set of $D$. We can solve \textsc{Hitting Set} in FPT time $d_U^{O(k)}n=d^{k^{O(k)}}n$~\cite{HW12}, which proves the above claim.

As mentioned before, if the obtained minimum-size absorbing set of $D$ has size at least~$k$, since it forms a valid multipacking, we can accept. Otherwise, $D$ can be covered by $k-1$ balls of radius at most $3k-2$. Each such ball has at most $\sum_{i=0}^{3k-2}(d-1)^i+1 = d^{O(k)}$ vertices, so in total $D$ has at most $n=d^{O(k)}$ vertices and a brute-force algorithm in time $n^{O(k)}$ is FPT.
\end{proof}

Next, we consider the vertex cover number, already considered for Theorem~\ref{thm:BD-VC}.

\begin{theorem}\label{thm:MP-VC}
\MP{} parameterized by the vertex cover number $c$ of the input digraph of order $n$ can be solved in FPT time $2^{2^{O(c)}}n^{O(1)}$.
\end{theorem}

\begin{proof}
Let $(D=(V,A),k)$ be the input of \MP{} and let $S$ be a vertex cover of $D$ of size $c$. As for Theorem~\ref{thm:BD-VC}, we partition the set $V\setminus S$ (which contains no arcs) into equivalence classes $C_1,\ldots, C_t$ according to their in- and out-neighborhoods in $S$. There are $t\leq 2^{2c}$ such classes.

By Lemma~\ref{lem:mp:sources}, we can assume that all sources belong to an optimal solution. Consider any class $C_i$. Its vertices are either all sources, or none of them are. If they are not sources, they all have a common in-neighbor, and thus at most one vertex of $C_i$ can belong to a multipacking. It is not important which one is selected, since all vertices in $C_i$ are twins. We may thus simply try all possibilities of selecting at most one vertex per class $C_i$, and all possibilities of selecting vertices of $S$. Thus, there are $2^{t+c}=2^{2^{O(c)}}$ potential multipackings of $D$ containing all sources. Each of them can be checked in polynomial time. This is an FPT algorithm.
\end{proof}

\section{Conclusion}\label{sec:conclu}

We have studied \BD{} and \MP{} on various subclasses of digraphs, with a focus on DAGs. It turns out that they behave very differently than for undirected graphs. 
We feel that \MP{} is slightly more challenging.

Indeed, we managed to solve some questions for \BD, that we leave open for \MP. For example, it would be interesting to see whether \MP{} is FPT for DAGs, and whether it remains W[1]-hard for digraphs without directed $2$-cycles. Also, \BD{} is FPT for nowhere dense graphs (as it can be expressed in first-order logic), however, it is not clear whether this holds for \MP{}. It is also unknown whether \MP{} is NP-hard on undirected graphs, as asked in~\cite{TeshimaMasterThesis,YangMasterThesis}.

On the other hand, we showed that \MP{} is NP-complete for single-sourced DAGs, but we do not know whether the same holds for \BD. %

We note that in most of our hardness reductions, the maximum finite distance is very small (which helps us to control the problems at hand), but the actual diameter is infinite (as our digraphs are not strongly connected). It seems a challenging question to derive hardness results for strongly connected digraphs, which can be seen as an intermediate class between the two extremes that are undirected graphs, and DAGs.

We have also shown that both problems are FPT when parameterized by the vertex cover number. 
What about smaller parameters such as tree-width or DAG-width? 

Finally, can our FPT algorithms for both problems parameterized by the solution cost/solution size and maximum out-degree be strengthened to a polynomial kernel?

\bibliographystyle{plain}
\bibliography{references}

\end{document}